\documentclass[a4paper,UKenglish,numberwithinsect,cleveref,thm-restate]{socg-lipics-v2021}

\usepackage[utf8]{inputenc}
\usepackage[T1]{fontenc} 

\usepackage{babel}
\usepackage{xspace}

\usepackage{algorithm}
\usepackage{algorithmic}

\usepackage{graphicx}
\usepackage{wrapfig}
\usepackage{xcolor}

\def\appendixSection{}         

\ifdefined\appendixFull
    \usepackage[bibliography=common]{apxproof}
    \def\appendixCommands{}
\else\ifdefined\appendixInline
    \usepackage[appendix=inline]{apxproof}
\else\ifdefined\appendixStrip
    \usepackage[appendix=strip]{apxproof}
\else\ifdefined\appendixSection
    \usepackage[appendix=chapterend,bibliography=common]{apxproof}
    \def\appendixCommands{}
    \makeatletter
    \newcommand\prepareAppendixSection{
        \def\axp@beginchapterappendix{\begingroup}
        \def\axp@endchapterappendix{\begingroup}
        \def\chapterappendixprelim{}
        \let\axp@section\subsection
        \let\section\subsection
        \let\subsection\subsubsection
    }
    \makeatother
\fi\fi\fi\fi

\makeatletter
\ifdefined\appendixCommands
    \newcommand{\thmToAppendix}[1]{\axp@writesection \immediate\write\axp@proofsfile{\noexpand\csname#1\endcsname*}}
    \newcommand{\appendixSubsection}[1]{\axp@writesection \immediate\write\axp@proofsfile{\noexpand\subsection{#1}}}
\else
    \newcommand{\thmToAppendix}[1]{}
    \newcommand{\appendixSubsection}[1]{}
\fi
\makeatother

\usepackage{tikz}
\usetikzlibrary{arrows.meta, positioning}

\crefname{property}{property}{properties}
\crefname{constraint}{constraint}{constraints}

\let\DeclareRobustCommandX\DeclareRobustCommand 
\newcommand{\NewMathVariable}[2]{\DeclareRobustCommandX{#1}[0]{\ensuremath{#2}\xspace}}
\newcommand{\NewProblem}[2]{\DeclareRobustCommandX{#1}[0]{\textsf{#2}\xspace}}

\NewMathVariable{\Nbb}{\mathbb{N}}
\NewMathVariable{\Qbb}{\mathbb{Q}}
\NewMathVariable{\Rbb}{\mathbb{R}}
\NewMathVariable{\Zbb}{\mathbb{Z}}
\NewMathVariable{\Rnn}{\mathbb{R}_{\geq 0}}
\NewMathVariable{\Znn}{\mathbb{Z}_{\geq 0}}
\NewMathVariable{\Zp}{\mathbb{Z}_{+}}

\usepackage{mathtools}
\DeclarePairedDelimiter{\paren}{(}{)}

\DeclarePairedDelimiter{\set}{\{}{\}}

\DeclarePairedDelimiter{\fceil}{\lceil}{\rceil}
\DeclarePairedDelimiter{\ffloor}{\lfloor}{\rfloor}

\makeatletter
\newcommand{\authorcommands}[4][blue!80!black]{%
    \expandafter\def\csname#2m\endcsname##1{\marginpar{$\leftarrow$\fbox{#4}}\footnote{$\Rightarrow$~{\sf ##1 --#3}}}
    \expandafter\def\csname #2\endcsname{\@ifstar{\csname @#2\endcsname}{\csname @@#2\endcsname}}
    \expandafter\def\csname @#2\endcsname##1{\textcolor{#1}{[#3: ##1]}}
    \expandafter\def\csname @@#2\endcsname##1{\par\noindent\textcolor{#1}{#3: ##1}\par}
}
\newcommand{\noauthorcommands}[4][blue!80!black]{%
    \expandafter\def\csname#2m\endcsname##1{}
    \expandafter\def\csname#2\endcsname##1{}
}
\makeatother

\pdfoutput=1
\hideLIPIcs
\nolinenumbers

\graphicspath{{./figures/}}

\title{Buy-at-Bulk Facility Location on Trees}

\author{Shamisa Nematollahi}
{Université Paris Cité, CNRS, IRIF, F-75013, Paris, France}
{shamisa@irif.fr}
{https://orcid.org/0009-0009-0486-8724}{}

\author{Daniel Vaz}
{LIGM, Université Gustave Eiffel, CNRS, ESIEE Paris, F-77454 Marne-la-Vallée, France}
{daniel.ramosvaz@esiee.fr}
{https://orcid.org/0000-0003-2224-2185}{}

\authorrunning{S. Nematollahi and D. Vaz} %
\Copyright{Shamisa Nematollahi and Daniel Vaz} %

\ccsdesc[500]{Theory of computation~Facility location and clustering}
\ccsdesc[500]{Theory of computation~Routing and network design problems}

\keywords{buy-at-bulk, facility location, PTAS, approximation algorithms}

\category{} %

\relatedversion{Conference version published at WAOA 2025, \url{doi: 10.1007/978-3-032-06706-7_12}}

\funding{
\textit{Shamisa Nematollahi:} This work was funded by ANR project ANR-21-CE48-0016 (COMCOPT).
\textit{Daniel Vaz:} This work was supported by ANR project ANR-21-CE48-0022 (S-EX-AP-PE-AL)}

\EventEditors{John Q. Open and Joan R. Access}
\EventNoEds{2}
\EventLongTitle{42nd Conference on Very Important Topics (CVIT 2016)}
\EventShortTitle{CVIT 2016}
\EventAcronym{CVIT}
\EventYear{2016}
\EventDate{December 24--27, 2016}
\EventLocation{Little Whinging, United Kingdom}
\EventLogo{}
\SeriesVolume{42}
\ArticleNo{23}

\renewcommand{\epsilon}{\varepsilon}
\DeclareMathOperator{\poly}{poly}
\DeclareMathOperator{\opt}{opt}
\DeclareMathOperator{\apx}{apx}
\DeclareMathOperator{\round}{round}
\DeclareMathOperator{\OPT}{OPT}

\DeclareMathOperator{\total}{total}
\DeclareMathOperator{\totalmin}{total_{\min}}

\DeclareMathOperator{\cost}{cost}

\DeclareMathOperator{\ex}{ex}
\DeclareMathOperator{\nondom}{nondom}
\DeclareMathOperator{\argmin}{argmin}

\newcommand{\ieps}[1][1]{\ensuremath{\epsilon^{-#1}}\xspace}
\newcommand{\dval}{\ensuremath{\lambda}\xspace}

\DeclareMathOperator{\len}{\ell}
\DeclareMathOperator{\ccap}{c_{\operatorname{cap}}}

\NewProblem{\bbfl}{BBFL}
\NewProblem{\kcfl}{kCFL}
\NewProblem{\icfl}{1CFL}
\NewProblem{\rakcfl}{RAkCFL}

\NewProblem{\st}{ST}
\NewProblem{\ufl}{UFL}

\NewProblem{\lp}{LP}
\NewProblem{\ptas}{PTAS}

\newcommand{\indic}{\ensuremath{\mathbf{1}}\xspace}
\newcommand{\zeroes}{\ensuremath{\mathbf{0}}\xspace}
\newcommand{\aset}{\ensuremath{\mathcal{A}}\xspace}
\newcommand{\asetplus}{\ensuremath{\mathcal{A}^+}\xspace}

\newcommand{\emu}{\ensuremath{\epsilon\mu}\xspace}

\renewcommand{\appendixsectionformat}[2]{Details of Section #1 (#2)}

\begin{document}
\maketitle
\nolinenumbers
\begin{abstract}
We consider the buy-at-bulk facility location problem (BBFL),
a problem combining the classic facility location problem with buy-at-bulk network design,
which finds motivation in telecommunication networks.
In it, we are given a graph with edge lengths, opening costs and demands for each vertex, and a monotone and subadditive capacity-cost function,
and our task is to open facilities on a subset of the vertices and route the demand from each vertex to these facilities.
The cost of a solution (which we want to minimize) is given by the opening costs of the chosen facilities, plus the cost on each edge, which is given by its length times the cost of providing enough capacity for the demands through the edge, given by the capacity-cost function.
A common variant of the problem, the $k$-cable facility location problem (kCFL),
considers the case where capacity is provided by buying copies of given cable types, each with a certain capacity and cost.

We study BBFL on tree instances and show, for the unit-demand and splittable variants, that the problem admits a PTAS (a $(1+\epsilon)$-approximation for any $\epsilon > 0$).
We also consider kCFL in the new setting of cable-unsplittable demands, where the demand of a vertex cannot be split among multiple cables.
We show that the problem is NP-hard to approximate to a factor better than $3/2$ on stars, and then provide an algorithm for tree instances that outputs a solution with optimal cost, but which exceeds the capacity on each cable by a factor of $1+\epsilon$.
As a consequence, we show that the problem has a $2$-approximation algorithm on trees.
\end{abstract}


\section{Introduction}

The facility location problem is a classic in operations research,
which concerns the placement of facilities to distribute goods or services
so that the costs of delivery and installation of the facilities remain low.
However, in cases such as telecommunication networks,
the structure of the distribution network matters as much as the actual distances,
as using common distribution trunks can vastly reduce the cost of distribution.
This has led to connection costs being considered from a network design perspective,
giving rise to a problem formulation that more closely considers its applications.

We consider one such formulation, the \emph{buy-at-bulk facility location (\bbfl)} problem, which combines the facility location and buy-at-bulk network design problems.
In it, we have \emph{demands} on the vertices of a graph, and we want to serve these demands using \emph{facilities}.
For this purpose, we must choose some vertices of the graph on which to open facilities and further install a network with enough capacity to support connecting all of the demands to an open facility.
The buy-at-bulk name comes from the fact that the network has subadditive costs:
the installation cost is a unit-length cost of supporting a certain amount of demand, which is a monotone, sub-additive function of the total demand on the edge.
For each edge, the cost is then given by the installation cost times the length of the edge.

Our goal in the buy-at-bulk facility location is to, given a graph with edge lengths, facility costs, vertex demands and an installation cost function, determine which facilities to open, and how to direct the demand of each vertex to an open facility, such that the total cost of open facilities and edges is minimized.

The most common version of the problem uses discrete cables to provide capacity, thus replacing the sub-additive installation cost function by the minimum cost of a configuration of cables that has capacity at least equal to the demand.
To distinguish this setting from the general subadditive case,
we call this version \emph{$k$-cable facility location (\kcfl)}.
The difference is formalized by adding $k$ cable types to the input, each with a capacity $\mu_i$ and a cost $c_i$.
When computing the solution, for each edge with a total demand of $\dval$, we must choose how many copies of each cable to buy, such that the total capacity is at least $\dval$; the installation cost is then the sum of the costs of each cable multiplied by the number of copies purchased.

The buy-at-bulk facility location problem was initially introduced by Meyerson et al.~\cite{MeyersonMP00}. They established that \bbfl is a specific instance of the more general Cost-Distance problem, and thus an $O(\log n)$-approximation algorithm follows (for unsplittable demands).
For \kcfl, Ravi and Sinha~\cite{RaviS02} improved the approximation factor to $O(k)$.
In the special case of a single cable type ($k=1$), which they call the capacitated-cable facility location problem,
they show an approximation factor of $\rho_{\ufl} + \rho_{\st}$ in the uniform-demand version of the problem,
and $2 \rho_{\ufl} + \rho_{\st}$ for non-uniform unsplittable demands,
where $\rho_{\ufl} = 1.488$~\cite{Li13/iandc} and $\rho_{\st} = 1.39$~\cite{byrka2011steiner} denote the best-known approximation ratios for the uncapacitated facility location and Steiner tree problems respectively.
\kcfl (for general $k$) and \bbfl are known to be NP-hard even on a single edge, as \kcfl models the knapsack problem~\cite{Lueker75}.

We focus on studying the approximability of the problem from the ground up,
starting with simple instances.
For this reason, this work focuses on tree graphs, and on optimizing the approximation ratio under different problem variants.
Trees are simple graphs, where techniques can first be studied to obtain near-optimal approximation algorithms (PTAS), and often provide techniques that generalize to other settings.
In particular, many network design and operations research problems are NP-hard on trees, with some having PTAS or good approximations, while others are hard to approximate even on trees.

We consider the \bbfl problem in the splittable demand setting,
in which the demand of each vertex can be split into parts that head to different facilities,
 and give a \ptas for the problem, which similarly implies a \ptas for \kcfl.

\begin{theorem}
\label{thm:intro:bbfl-split:ptas}
There is a \ptas for the \bbfl problem on trees with splittable demands that runs in time $f(\epsilon) \cdot \poly(n)$.
\end{theorem}

\begin{corollary}
\label{cor:intro:kcfl-split:ptas}
There is a \ptas for the \kcfl problem on trees with splittable demands that runs in time $f(\epsilon) \cdot \poly(n,k)$.
\end{corollary}

We then move on to studying \kcfl under a setting that is, as far as we know, novel:
using the formulation of the problem as buying cables, we consider the \emph{cable-unsplittable} setting, where the demands passing through an edge must be partitioned into the cables that are installed on that edge, that is, each demand must be assigned to a single cable on each edge, and cannot be split into multiple cables.
This is a natural setting when we consider that the ``cables'' in the problem can represent trucks or other discrete containers,
and a demand, which can represent a package, is indivisible and thus confined to a single container.
In these scenarios, repacking or splitting demands at intermediate nodes is either operationally expensive, technologically infeasible, or not allowed by design.

We show that this problem is APX-hard, as it cannot be approximated to a factor better than $3/2$, even when the input graph is a star.

\begin{theorem}
\label{thm:intro:kcfl-unsplit:apx-hardness}
\kcfl is APX-hard when the input graph is a star and $k=1$.
In particular, it is NP-hard to approximate the problem to a factor of $3/2-\epsilon$ for any $\epsilon > 0$.
\end{theorem}

We then show that, by allowing our solution to slightly overload the cables (exceeding the capacity by a factor of $1+\epsilon$), we can compute a solution with optimum cost:

\begin{theorem}
\label{thm:intro:kcfl-unsplit:ra-ptas}
For any $\epsilon>0$,
there is an algorithm for the \kcfl problem on trees with cable-unsplittable demands that computes a solution with optimum cost but exceeds cable capacities by a factor of at most $1+\epsilon$; the algorithm runs in time $n^{k\cdot f(\epsilon)}$.
\end{theorem}

As a consequence, we show that we can obtain a $2$-approximation for the problem:

\begin{theorem}
\label{thm:intro:kcfl-unsplit:apx}
There is a 2-approximation algorithm for the \kcfl problem on trees with cable-unsplittable demands that runs in time $n^{O(k)}$.
\end{theorem}

We start by presenting the formal problem definitions, as well as the connection between \bbfl and \kcfl in \Cref{sec:prelim}.
We then present the results for the splittable demand setting in \Cref{sec:split}, starting with a warm-up on paths to better explain the ideas of the algorithm (\Cref{sec:split:paths}).
Finally, we present the results for the cable-unsplittable setting of \kcfl in \Cref{sec:unsplit}, which for simplicity we consider in the single-cable setting (\Cref{sec:unsplit:k1}); then we give the details on multiple cable types in \cref{sec:unsplit:k}.

We present a short technical overview at the beginning of each section, in particular \Cref{sec:split,sec:split:paths} for the first setting and \Cref{sec:unsplit,sec:unsplit:k1} for the second.
We defer technical proofs to \Cref{sec:proofs}, though we try to give a sketch of the arguments when important.

\subsection{Related Work}

Buy-at-bulk connected facility location is a closely related variant of the problem where open facilities are required to be connected. In this setting, we are given an extra type of cable, called core cable, with a large enough installation cost of $M$ and capacity $\mu = +\infty$, which is used to connect the open facilities.
Friggstad et al.\ gave an \lp-based constant factor approximation algorithm~\cite{FriggstadRSS19}.

Another closely related problem is the airport and railway (AR) problem, corresponding to a variant of the problem where the solution is a forest and facilities have capacity limits (i.e., each facility can serve at most $k$ units of demand).
The AR problem has been studied both with splittable demands \cite{AdamaszekA0M18,AdamaszekAM16,SalavatipourT24} and unsplittable demands~\cite{JowhariN25} in various settings.
Salavatipour et al.~\cite{SalavatipourT24} gives an $O(\log n)$-approximation to the problem, which improves to a $2$-approximation with uniform facility costs.

Buy-at-bulk costs are ubiquitous in our society, and were first introduced
to network design by Salman et al.~\cite{SalmanCRS97}.
Since then, different variants of buy-at-bulk network design have received considerable attention~\cite{Andrews04,AwerbuchA97,ChekuriHKS10,ChekuriKN01,GrandoniI06,GuhaMM01,GuptaKR03,Talwar02}, with an $O(\log n)$-approximation for the multi-commodity setting being the most well-known result~\cite{AwerbuchA97,FakcharoenpholRT04}.

There is a long line of work on the metric facility location problem, both in the capacitated~\cite{AggarwalLBGGGJ13,AnSS17,BansalGG12,
BehsazSS16,ChudakW05,KorupoluPR00,PalTW01} and uncapacitated variants~\cite{Byrka07,GuhaK98,JainMMSV03,JainMS02,Li13/iandc,MahdianYZ06,ShmoysTA97,Sviridenko02}.
For the capacitated variant, the best approximation has factor $5$~\cite{BansalGG12}, and an improved $3$-approximation exists when capacities are uniform~\cite{AggarwalLBGGGJ13}.
For the uncapacitated variant, a $1.488$-approximation algorithm by Li~\cite{Li13/iandc} nearly matches the inapproximability of $1.463$~\cite{GuhaK98,Sviridenko02}.
For trees, the problem is known to be polynomial-time solvable (see e.g.~\cite{Shaw99/anor}) in both capacitated and uncapacitated settings.


\newcommand{\bidir}[1]{\ensuremath{\tilde{#1}}}

\section{Preliminaries}
\label{sec:prelim}

Throughout our work, we use standard graph notation (see e.g.\ \cite{Diestel/0030488/daglib}).
We consider every tree $T$ as having a root $r$, and define parent and child of a vertex $v$ with respect to $r$.
For a vertex $v$, the subtree rooted in $v$ is denoted $T_v$.

For an edge set $E$, we denote by $\bidir{E}$ the set of arcs (directed edges) containing both directions of an edge, that is,
\(
    \bidir{E} = \bigcup_{uv \in E} \set{(u,v), (v,u)},
\)
and denote by $\delta^-(v)$ and $\delta^+(v)$ the incoming and outgoing arcs, respectively, of a vertex $v$ in $\bidir{E}$.

For $x, y \in \Zbb$, let $[x, y] = \set{z \in \Zbb :x \leq z \leq y}$, $[x] = [1,x]$.
We use $\indic^n(S)$ to refer to the indicator vector of $S\subseteq[n]$ in $\set{0,1}^n$, that is, a vector $x \in \set{0,1}^n$ such that $x_i = 1$ if $i \in S$ and $x_i = 0$ otherwise; we write $\indic^n(v)$ instead of $\indic^n({\set{v}})$ and omit the superscript if clear from context.
We use $\zeroes^n \in \set{0}^n$ to be the all-zeroes vector, and write $\zeroes$ if $n$ is clear from context.
Given sets $Y \subseteq X$ and a function $f\colon X \to \Rnn$, we denote $f(Y) = \sum_{y \in Y} f(y)$.

In our algorithms, we sometimes say that one solution \emph{dominates} another to mean that the first is at least as good as the second.
For instance, let $(\mu,c)$ represent a cable, where $\mu$ is the capacity and $c$ is the cost.
We might say that a cable $(\mu, c)$ dominates another cable $(\mu', c')$ if $\mu\geq \mu'$ and $c \leq c'$, as it has at least as much capacity and costs at most as much.
Though the definition of ``domination'' depends on the context and will be introduced as needed, we always say that a solution is \emph{non-dominated} within a set $A$ if no different solution dominates it,
and we denote the \emph{set of non-dominated solutions} in $A$ as $\nondom(A)$.

\subsection{Problem Definitions}
\label{sec:prelim:problem}

We consider two problems,
buy-at-bulk facility location (\bbfl) and $k$-cable facility location (\kcfl),
which blend buy-at-bulk network design and facility location~\cite[Sec.~4.5, Sec.~8.6]{WilliamsonS11/0030297/daglib}.
This section first explains these two settings, how they combine into \bbfl, and how the concept of cables leads to \kcfl.

\subparagraph{Buy-at-bulk network design.}

Network design problems include many well-known combinatorial optimization problems, such as minimum spanning tree and TSP.
At its core, the goal is to find a minimum-cost subgraph that satisfies given demand constraints, where the simplest version is to provide sufficient capacity to serve point-to-point demands.

The buy-at-bulk network design problem takes as input a graph $G=(V,E)$ with edge lengths $\len\colon E \to \Rnn$, a non-decreasing subadditive capacity cost function $\ccap\colon \Znn \to \Rnn$ and source-sink pairs $(s_i, t_i, d_i)$ with demand $d_i$.
The goal is to choose an $s_i$-$t_i$-path $P_i$ for each $i$, such that the total cost on the edges is minimized, where the cost on an edge depends on the total demand through it, $\dval_e = \sum_{P_i \ni e} d_i$, and is given as the product of the length and capacity cost for the demand, for a total of
\(
    \sum_{e \in E} \len(e) \cdot \ccap(\dval_e).
\)

The assumption that $\ccap$ is non-decreasing and subadditive captures the \emph{buy-at-bulk} principle: as demand grows, the cost of installing capacity increases, but the cost per unit demand decreases.
Formally, subadditivity means that $\ccap(x + y) \leq \ccap(x) + \ccap(y)$ for all $x, y \in \mathbb{Z}_{\geq 0}$, reflecting the fact that it is more economical to install a single larger capacity to serve combined demand than to install multiple smaller ones separately.

In our problems, we consider \emph{set-to-set} demands of the form $(X_i, Y_i, d_i)$, where the solution must simultaneously support a flow of value $d_i$ from $X_i$ to $Y_i$ for each demand.

\begin{definition}
Let $G$ be a graph and $\mu\colon E \to \Rnn$ be edge capacities.

A flow $g\colon \bidir{E} \to \Rnn$ is a function that satisfies edge capacities, that is, $g(e) \leq \mu(e)$ for every edge $e \in E$.

The \emph{excess} of $g$ at a vertex $v \in V(G)$ is defined as
$
\ex_g(v) := \sum_{e \in \delta^-(v)}g(e) - \sum_{e \in \delta^+(v)}g(e).
$

We say that a flow is an \emph{$S$-$T$-flow}, or \emph{from $S$ to $T$}, if $\ex_g(v) \geq 0$ for $v \in T$, $\ex_g(v) \leq 0$ for $v \in S$, and $\ex_g(v) = 0$ otherwise.
Its \emph{value} is defined as $|g| := \sum_{v \in T} \ex_g(v)$.
\end{definition}

\subparagraph{Facility location.}

In the facility location problem, we are given a set of clients $V$ and a set of facilities $F$, as well as connection costs $c_{ij}$ of assigning a client $j \in V$ to a facility $i \in F$, and opening costs $f\colon F \to \Rnn$ for each facility.
The goal of this problem is to pick a subset of the facilities $I \subseteq F$ and an assignment $\sigma\colon V \to I$ such that we minimize the total opening and connecting costs,
\(
	\sum_{i \in I} f(i) + \sum_{j \in V} c_{\sigma(j)j}.
\)

\subsubsection{The Buy-at-Bulk Facility Location Problem}

The buy-at-bulk facility location problem (\bbfl) is obtained by combining the above as follows:
we start with a facility location problem on a graph $G$, where both facilities and clients correspond to the vertices of $G$;
each vertex also has a demand, which must be served by one or more facilities, and thus the solution is a subgraph that supports flows from each vertex to the opened facilities with value equal to the demand;
for this purpose, capacities on the edges are not fixed but can be purchased as in a buy-at-bulk setting, with an installation cost function dictating the unit-length cost of providing capacity for a certain demand.

Formally, the input is a tuple $(G,\len,f,d,\ccap)$ with the following parts:
\begin{itemize}
    \item a graph $G=(V,E)$ with edge lengths $\len\colon E\to \Rnn$;
    \item opening costs $f\colon V\to \Rnn$;
    \item demands $d\colon V \to \Znn$;
    \item a monotone and sub-additive installation cost function $\ccap\colon \Znn \to \Rnn$.
\end{itemize}

\noindent
A solution is a tuple $(I,\mathbf{g})$ consisting of:
\begin{itemize}
	\item A set $I \subseteq V$ of facilities to open;
	\item For each vertex $v \in V$, a flow $g_v$ of value $d(v)$ from $\set{v}$ to $I$ .
\end{itemize}

\noindent
The goal of the problem is to find a solution with total minimum cost, given by
\[
\sum_{i\in I} f(i) + \sum_{e\in E}\len(e)\cdot \ccap(g'(e)),
\]
where $g'(uw) := \paren[\big]{\sum_v g_v(u,w)+g_v(w,u)}$ is the \emph{demand on the edge} $uw$.

It is particularly important for our algorithms that a $(1+\epsilon)$-approximation $y$ for $\ccap(x)$ ($\ccap(x) \leq y \leq (1+\epsilon)\ccap(x))$ be computable in polynomial-time.

Let $(I, \set{g_v}_v)$ be a solution.
We say that a demand of $\delta$ \emph{passes} an edge $uw$ (in the direction $(u,w)$) if for some $v \in V$, $g_v(u,w) = \delta$, that is, the flow of $v$ on $(u,w)$ is $\delta$.
Similarly, we say that a demand of $\delta$ \emph{leaves} $u$ if for some $v \in V$, a flow of $\delta$ leaves $u$, $g_v(\delta^+(u)) = \delta$.
We denote by $\cost(uv, \dval) =  \len(uv) \cdot \ccap(\dval)$ the cost of sending $\dval$ units of demand through edge $uv$.

\subsubsection{The k-Cable Facility Location Problem}

The $k$-cable facility location problem is a special case of \bbfl where capacity is obtained by buying a set of cables, each of which has a cost and a capacity.
In particular, we have $k$ available cable types, each with a capacity $\mu_i$ and cost $c_i$, and for each edge we buy copies of each cable so that the total capacity is sufficient to cover the demand.

In the literature, it is generally assumed that the cables satisfy {economies of scale}, meaning that as capacity increases, cost increases, but the cost-per-capacity ratio decreases.
While this assumption is reasonable, it is not required in our algorithms.

We remark that \kcfl is a special case of \bbfl, since the function is
monotone, as covering a larger total demand cannot be cheaper, and
subadditive, as covering the sum of two demands can be done by putting together their cable sets.

An instance of \kcfl is a tuple $(G,\len,f,d,\mathbf{q})$, where:
\begin{itemize}
    \item a graph $G=(V,E)$ with edge lengths $\len\colon E\to \Rnn$;
    \item opening costs $f\colon V\to \Rnn$;
    \item demands $d\colon V \to \Znn$;
    \item $k$ cable types $\mathbf{q}=\{(\mu_1, c_1), (\mu_2, c_2),\ldots , (\mu_k, c_k)\}$, each with capacity $\mu_i$ and cost $c_i$.
\end{itemize}

\noindent
A solution is a tuple $(I,\mathbf{g},\mathbf{Q})$, composed of:
\begin{itemize}
    \item A set $I \subseteq V$ of facilities to open;
    \item For each vertex $v \in V$, a flow $g_v$ of value $d(v)$ from $\set{v}$ to $I$;
    \item For each edge $uw \in E$, a bag $Q_{uw}$ that has sufficient capacity to support the flows $\set{g_v}_v$, i.e.\ $\mu(Q_{uw}) \geq g'(uv) := \paren[\big]{\sum_v g_v(u,w)+g_v(w,u)}$.
\end{itemize}

\noindent
The objective value of a solution is:
\[
\sum_{i\in I} f(i) + \sum_{e\in E}\len(e)\cdot c(Q_e).
\]

Given a set $U$, we denote by $Q \subseteq U^*$ a multi-set containing elements of $U$ with arbitrary repetition, and refer to a multi-set of cables $Q \subseteq \mathbf{q}^*$ as a \emph{bag}.
The capacity of a bag $Q$ is $\mu(Q) = \sum_{i \in Q} \mu_i$ and its cost is $c(Q) = \sum_{i \in Q} c_i$.

The lemma below shows that we can
compute a $(1+\epsilon)$-approximation to the minimum-cost bag supporting a given demand $\dval$, as required by the definition of \bbfl.
The proof follows standard techniques for the knapsack problem~\cite[Sec.~3.1]{WilliamsonS11/0030297/daglib}.

\begin{lemma}[restate=lemCableOracle,name=*]
Let $\{(\mu_i, c_i)\}_{i \in [k]}$ be a set of $k$ cables and $\dval$ be a target demand.

For $\epsilon > 0$, there exists an $(1+\epsilon)$-approximation algorithm to find a minimum-cost bag that has total capacity at least $\dval$, and which runs in time $\poly(k, 1/\epsilon)$.
\end{lemma}

\thmToAppendix{lemCableOracle}

\begin{appendixproof}
We leverage a dynamic programming approach inspired by the PTAS for the 0-1 minimum knapsack problem to find a near-optimal solution for this cable selection problem.

We use the standard framework in the book by Williamson and Shmoys~\cite[Sec.~3.1]{WilliamsonS11/0030297/daglib}.
Given two pairs $(\mu, c)$ and $(\mu', c')$, where $\mu$ and $\mu'$ represent capacity, and $c$ and $c'$ represent cost, we say that $(\mu, c)$ \emph{dominates} $(\mu', c')$ if $\mu \geq \mu'$ and $c \leq c'$,

We will maintain a list of solutions $A[i]$ for each $i \in [k]$,
corresponding to the best trade-offs between capacity and cost that can be achieved using the first $i$ cables.
Each entry of $A[i]$ is a pair $(\mu, c)$, corresponding to the capacity $\mu$ and cost $c$ of a bag using the first $i$ cables.
As in the mentioned framework~\cite[Sec.~3.1]{WilliamsonS11/0030297/daglib}, we keep only the entries of $A[i]$ that are non-dominated.

$A[i]$ is computed as follows:
\begin{itemize}
    \item  $A[1] = \{(0, 0),(\mu_1, c_1)\},$
    \item  $A[i] = \nondom\set[\big]{(\mu + j \cdot \mu_i , c + j\cdot c_i):(\mu, c)\in A[i-1], 0\leq j\leq \fceil*{\paren{\lambda-\mu}/{\mu_i}} }$
\end{itemize}

After computing $A[k]$, find the pair $(\mu, c)$ in $A[k]$ such that $\mu \ge \lambda$ and $c$ is minimized. This pair represents the $(1+\epsilon)$-approximate minimum-cost cable selection.

We represent a solution as a vertex $x\in \mathbb{Z}^K_{\geq 0}$ counting the number of copies of each cable.

The algorithm above runs in time $O(k\cdot C)$, where $C$ is the number of different possible values for the cost.
To obtain an algorithm that runs in time $\poly(k, 1/\epsilon)$, we round the costs up to multiples of  $\delta:=\frac{\epsilon}{2k}c_{\apx}$,
i.e.~$c'_i = \fceil{c_i/\delta} \cdot \delta$,
where $c_{\apx}$ is a $2$-approximation to the optimum cost, and solve the solution using the rounded costs.

To obtain such a $2$-approximation, we construct (at most) two solutions and take the cheapest:
the first is a minimum-cost cable with capacity at least $\dval$,
and for the second, we consider the cable with capacity at most $\dval$ that minimizes the ratio $c_i/\mu_i$ and take the cost of $j:= \fceil*{\dval/\mu_i}$ copies of the cable, for a cost of $j \cdot c_i$.

An optimum solution of cost $c_{\opt}$ either uses a cable of capacity at least $\dval$, in which case it costs at least as much as the first solution,
or otherwise $c_{\opt}/\dval \geq c_i/\mu_i$, as the solution uses cables with ratio at least equal to $c_i/\mu_i$ (by minimality of $i$).
In this second case, we can bound the cost as
\[
j \cdot c_i
= \fceil*{\frac{\dval}{\mu_i}} \cdot c_i
\leq 2\frac{\dval}{\mu_i} \cdot c_i
\leq 2 \dval \frac{c_{\opt}}{\dval} = 2c_{\opt},
\]
where we use in the first inequality that $\mu_i \leq 2\dval$ and $\fceil{x} \leq 2x$ for $x \geq 1/2$,
and in the second inequality the fact that $c_{\opt}/\dval \geq c_i/\mu_i$.

\subparagraph{Approximation Factor Analysis.} Let $c_{\opt}$ be the optimal cost using the original costs, and let $c_{\round}$ be the cost of the solution found by the dynamic program using the rounded costs. We want to show that $c_{\round} \le (1+\epsilon)c_{\opt}$.

For each cable $i$, the rounded cost $c_i'$ is at most $c_i + \delta$. Therefore, $c_i' \le c_i + \delta$.

Let $S$ be the set of cables chosen by the algorithm, and  $S_{\opt}$ be the set of cables in the optimal solution. The cost of the solution with rounded costs is $c_{\round} = \sum_{i \in S} c_i'$, and the same solution has cost $c = \sum_{i \in S} c_i$ with original costs. We remark that $c \leq c_{\round}$, since $c_i\leq c_i'$ for each item.

Also, since the solution found is optimal for the rounded costs, and the optimal solution with original costs has a corresponding solution with rounded costs, we conclude that:
\[
c \leq c_{\round} \leq \sum_{i \in S_{\opt}} c_i' \leq \sum_{i \in S_{\opt}} (c_i + \delta) \leq c_{\opt} + k \cdot \delta \leq c_{\opt} + k \cdot \frac{\epsilon}{2k} c_{apx} \leq c_{\opt} + \epsilon \cdot c_{\opt}\,.
\]
Since $c_{apx} \leq 2c_{\opt}$, the solution with rounded costs is a $(1+\epsilon)$-approximation of the optimal solution.
\end{appendixproof}

\subsubsection{Variants of \bbfl and \kcfl}

Both of the problems above can be considered in different variants, depending on whether we allow the demands to be split along different paths.
The following classic variants are usually considered:
\begin{itemize}
	\item {Unit demand:} $d(v) = 1$ for each $v \in V$;
	\item {Unsplittable:} $g_v$ is supported on a single path for each $v \in V$;
	\item {Splittable:} No restrictions.
\end{itemize}

Notice that in the unit-demand case, any flow is supported on a path w.l.o.g.\ by the integral flow theorem, since capacities are integers (see e.g.~\cite[Corollary 8.7]{KorteV2018}).

Depending on the setting, we might omit some parts of the solution if there is a clear, optimal choice.
For instance, in the unit demand case when the input graph is a tree, the flow $g_v$ is always of value $1$ and determined by the facility that serves the demand of $v$.

\subparagraph{Cable-unsplittable demands.}
We further introduce the \emph{cable-unsplittable} variant of \kcfl, where the demand of a single vertex, besides being supported on a single path, has to be assigned to a specific cable on each edge.
In other words, for each edge $e$, we ask for a bag $Q_e$ such that the demands passing through $e$ are partitioned onto the cables of $Q_e$, with the total demand assigned to each cable not exceeding its capacity.
We consider that cables are uni-directional, that is, the demands assigned to each cable flow in the same direction.

Formally, $Q_{uw}$ is partitioned (according to direction) into $Q_{(u,w)}$ and $Q_{(w,u)}$,
and we have $h_{(u,w)}\colon V \to Q_{(u,w)}$, $h_{(w,u)}\colon V \to Q_{(w,u)}$,
such that for each $q=(c,\mu) \in Q_{uw}$, $\sum_{v:h_{(u,w)}(v)=q} g_v(u,w) \leq \mu$ and $\sum_{v:h_{(w,u)}(v)=q} g_v(w,u) \leq \mu$.
We do not explicitly state the partition of $Q_{uw}$ or the assignments $h_{(u,w)}$, $h_{(w,u)}$ as part of the solution unless needed.

As far as we know, this is the first use of this variant, and also the only variant of \kcfl that does not easily reduce to \bbfl.

\section{Splittable-Demand \bbfl}
\label{sec:split}

In this section, we show a PTAS for the \bbfl problem on trees with splittable demands.
For a simpler introduction to the core ideas of our dynamic program, \Cref{sec:split:paths} describes the case of path instances and polynomial demands.

The following lemmas are useful when designing dynamic programs for the problem:

\begin{lemma}[restate=lemUncrossing,name=*]
\label{lem:uncrossing:edges}
In the context of splittable (or unit) demands and a monotone cable cost function, there is an optimal solution without crossings; that is, for any edge $uw$, demand cannot pass in the directions $(u,w)$ and $(w,u)$ simultaneously.
\end{lemma}

\thmToAppendix{lemUncrossing}

\begin{appendixproof}
Suppose no optimal solution exists with no crossing edges,
and let $(I, \set{g_v}_v)$ be an optimum solution with the least crossings,
where $I$ is the set of open facilities, and each $g_v$ is a flow from $\set{v}$ to $I$ with value $d_v$.
Formally, we consider a crossing as an edge $uw \in E$ together with a pair of flows $g_a$, $g_b$ that have positive value on $(u,w)$ and $(w,u)$, respectively, and we take $(I, \set{g_v}_v)$ that has the minimum number of crossings.

Let $uw \in E$, $g_a$, $g_b$ be a crossing.
We will uncross the flows $g_a$, $g_b$ as long as there are paths using $e$ with positive flow to an open facility,
namely path $P_a$ from $a$ (to $I$) containing $(u,w)$,
and path $P_b$ from $b$ (to $I$) containing $(w,u)$.

Let $\delta_a = \min\set{g_a(e): e \in P_a}$ and $\delta_b = \min\set{g_b(e): e \in P_b}$ be the flow along those paths, and let $\delta_b \leq \delta_a$ w.l.o.g.
We modify the solution as follows:
\begin{itemize}
    \item Decrease $g_a(u,w)$ by $\delta_b$ to $g_a(u,w):=\delta_a-\delta_b$;
    \item Decrease $g_b(w,u)$ by $\delta_b$ to $g_a(u,w):=0$;
    \item For every edge $e$ of $P_a$ after $w$, decrease $g_a(e)$ by $\delta_b$ and increase $g_b(e)$ by $\delta_b$;
    \item For every edge $e$ of $P_b$ after $u$, decrease $g_b(e)$ and increase $g_a(e)$ by $\delta_b$.
\end{itemize}

Notice that we decrease flow $g_a$ along a path from $u$ to an open facility ($P_a$), and increase it by the same amount along a different path from $u$ to an open facility ($P_b$).
Similarly for $g_b$, we decrease and increase the flow along paths from $w$ to an open facility.
Thus, flow conservation is preserved, and excess remains positive at the open facilities.
The total flow and capacity do not increase on any edge, as the increases on $g_a$ are compensated by decreases on $g_b$ (or vice-versa) in every edge,
with the exception of $uw$, where the flow decreases by $2\delta_b$ overall.

After sufficient repetition over different paths in the support of $g_a$, $g_b$, the solution no longer has a crossing on $uw$ for these flows, and thus the number of crossings of the new solution is one fewer than the original, contradicting the minimality of the number of crossings.
Thus, an optimal solution with no crossings must exist.
\end{appendixproof}

\begin{lemma}[restate=lemFacilities,name=*]
\label{lem:uncrossing:facilities}
In the setting of uncapacitated facilities,
there is an optimal solution without facility-crossings; that is, no demand leaves an open facility.
\end{lemma}

\thmToAppendix{lemFacilities}

\begin{appendixproof}
Consider an optimal solution. For any demand passing an open facility $v$, we modify the solution so that it is served by $v$ instead. Since the facility $v$ is already open, this incurs no additional cost, and the demand on edges only decreases. Thus, the modified solution is optimal.
\end{appendixproof}

\begin{observation}[restate=obs:uncrossing:unsplit, name=]
\Cref{lem:uncrossing:edges} does not hold for unsplittable demands (see \Cref{ex:uncrossing_failure}).
\end{observation}

\thmToAppendix{obs:uncrossing:unsplit}

\begin{toappendix}
\begin{example}
\label{ex:uncrossing_failure}
Consider the following instance of the facility location problem with buy-at-bulk costs on a path graph \(P = (v_1, v_2, v_3, v_4, v_5)\) as represented in Figure~\ref{fig:optimal_solution_example}, where:
\begin{itemize}
    \item Facility opening costs are \(f(v_1) = f(v_5) = 1\), and \(f(v_2) = f(v_3) = f(v_4) = + \infty\);
    \item Vertex demands are \(d(v_2) = 2\), \(d(v_3) = 5\), \(d(v_4) = 3\), and \(d(v_1) = d(v_5) = + \infty\) (used to ensure that facilities must be opened at both ends);
    \item Edge lengths are:\quad
    $
    \ell(v_1,v_2) = 101,\quad \ell(v_2,v_3) = 1,\quad \ell(v_3,v_4) = 3,\quad \ell(v_4,v_5) = 100;
    $
    \item There are two cable types:\quad
    $
    \text{Type 1: } (\mu_1, c_1) = (3, 4), \quad
    \text{Type 2: } (\mu_2, c_2) = (5, 5).
    $
\end{itemize}

We show that in the optimal solution, demand must traverse the edge \((v_2, v_3)\) in \emph{both directions}, violating the condition in the uncrossing lemma.

Open facilities at \(v_1\) and \(v_5\). Then:
\begin{itemize}
    \item Route the demand of \(v_2\) to \(v_5\) using:
         cables of type 1 on edges \((v_2,v_3)\) and \((v_3,v_4)\),
         and of type 2 on edge \((v_4,v_5)\),
    for a total cost of \(4 \cdot 1 + 4 \cdot 3 + 5 \cdot 100 = 516\).

    \item Send the demand of \(v_4\) to \(v_5\) without additional cost, since the existing cable on \((v_4,v_5)\) has sufficient remaining capacity.

    \item Send the demand of \(v_3\) to \(v_1\) by installing:
        cables of type 2 on edges \((v_2,v_3)\) and \((v_1,v_2)\),
    for a cost of \(5 \cdot 1 + 5 \cdot 101 = 510\).
\end{itemize}

The total cost of the solution is:
\[
f(v_1) + f(v_5) + \text{installation cost} = 1 + 1 + 516 + 510 = 1028.
\]

Note that in this optimal solution, the edge \((v_2, v_3)\) carries demand in both directions simultaneously, as
    the demand of \(v_2\) travels rightwards from \(v_2\) to \(v_5\), and
    the demand of \(v_3\) travels leftwards from \(v_3\) to \(v_1\).

To see that this solution is optimal, notice that the total demand is 10, so any solution needs cost at least 1000 to reach $v_1$ or $v_5$ (with length at least 100).
Furthermore, an optimal solution must use two cables of type 2 on the edges incident to $v_1$ and $v_5$ (in total), as otherwise the cost will be higher.
Indeed, since the cost per capacity for type 1 is higher than for type 2, any other configuration of cables will have cost at least 11, which at length 100 gives cost at least 1100, which cannot be optimal.
As the demand of $v_3$ completely uses the capacity of a cable of type 2, it can be routed completely independently to the cheapest facility, $v_1$ (cost $1+5\cdot 102$ vs $5 \cdot 103$ for $v_5$).
The remaining demands are then routed as cheaply as possible, which must be to $v_5$.

This shows that the uncrossing lemma does not extend to unsplittable demands.

\end{example}
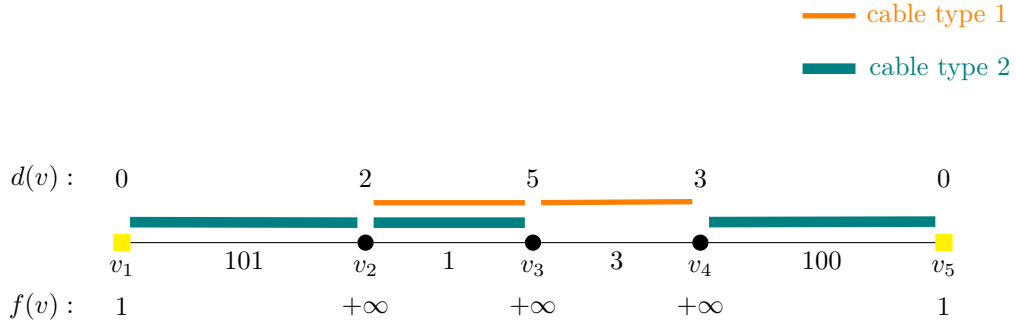
\begin{figure}[t]
	\centering
	\begin{tikzpicture}
		\tikzset{
			tealrect/.style={
				rectangle,
				draw=yellow,
				fill=yellow,
				thick,
				minimum size=0.2cm,
				inner sep=2pt,
				font=\small,
				align=center
			},
			vertex/.style={
				circle, draw=black, fill=black, minimum size=0.2cm, inner sep=0pt
			}
		}

		\node[tealrect,label=below:$v_1$] (v1) {};
		\node[vertex, draw=white , fill= white, above=0.05cm of v1] (v11) {};
		\node[vertex,right=3cm of v1,label=below:$v_2$] (v2) {};
		\node[vertex, draw=white , fill= white, right=3cm of v1, above=0.05cm of v2] (v22) {};
		\node[vertex, draw=white , fill= white, above=0.05cm of v22] (v222) {};
		\node[vertex,right=2cm of v2,label=below:$v_3$] (v3) {};
		\node[vertex,draw=white , fill= white, right=2cm of v2, above=0.05cm of v3] (v33) {};
		\node[vertex, draw=white , fill= white, above=0.05cm of v33] (v333) {};
		\node[vertex,right=2cm of v3,label=below:$v_4$] (v4) {};
		\node[tealrect,draw=white , fill= white, above=0.05cm of v4] (v44) {};
		\node[vertex, draw=white , fill= white, above=0.05cm of v44] (v444) {};
		\node[tealrect,right=3cm of v4,label=below:$v_5$] (v5) {};
		\node[tealrect,draw=white , fill= white, above=0.05cm of v5] (v55) {};

		\node[below=5mm of v1] (f1) {\(1\)};
		\node[left=3mm of f1] {\(f(v):\)};
		\node[below=5mm of v2] {\(+\infty\)};
		\node[below=5mm of v3] {\(+\infty\)};
		\node[below=5mm of v4] {\(+\infty\)};
		\node[below=5mm of v5] {\(1\)};

		\node[above=5mm of v1] (d1) {\(0\)};
		\node[left=3mm of d1] {\(d(v):\)};
		\node[above=5mm of v2] {\(2\)};
		\node[above=5mm of v3] {\(5\)};
		\node[above=5mm of v4] {\(3\)};
		\node[above=5mm of v5] {\(0\)};

		\draw (v1) -- node[below, midway] {\(101\)} (v2);
		\draw (v2) -- node[below, midway] {\(1\)} (v3);
		\draw (v3) -- node[below, midway] {\(3\)} (v4);
		\draw (v4) -- node[below, midway] {\(100\)} (v5);


		\draw[-, line width=2pt, orange] (v222) -- (v333);
		\draw[-, line width=2pt, orange] (v333) -- (v444);
		\draw[-, line width=4pt, teal] (v44) -- (v55);

		\draw[-, line width=4pt, teal] (v33) -- (v22);
		\draw[-, line width=4pt, teal] (v22) -- (v11);

		\begin{scope}[xshift=9cm,yshift=3cm]
			\draw[-, line width=2pt, orange] (0,0) -- (0.7,0) node[right] {cable type 1};
			\draw[-, line width=4pt, teal] (0,-0.7) -- (0.7,-0.7) node[right] {cable type 2};
		\end{scope}

	\end{tikzpicture}
	\caption{\nolinenumbers Optimal solution of Example~\ref{ex:uncrossing_failure}.}
	\label{fig:optimal_solution_example}
\end{figure}

\end{toappendix}

\subsection{Splittable-Demand \bbfl on Paths}
\label{sec:split:paths}
\appendixSubsection{Splittable-Demand BBFL on Paths}

Our algorithm is a dynamic program that computes the cost of a solution in a subpath under certain demand constraints.
Although our formulation is not the most natural for paths, it can generalize to trees without too much effort.
The algorithm in this section runs in time $O(n \cdot D)$, where $D = \sum_{v\in V}d(v)$.

Let $(P, \len, f, d, \ccap)$ represent an instance of the problem as defined in \Cref{sec:prelim:problem}, where $P$ is a path.
We consider $P$ as a tree and the two ends as the root $r$ and the leaf $w$.

Let $v$ be a vertex and $\dval \in [-D,D]$.
We use our dynamic program to compute the minimum cost of serving the
sub-path $P_{vw}$, subject to the demand constraint \dval as follows:
\begin{itemize}
    \item If $\dval\geq 0$, a demand of at most $\dval$ leaves node $v$ to its parent.
    \item If $\dval< 0$, a demand of at most $|\dval|$ arrives at node $v$ from its parent.
\end{itemize}

\subparagraph{Dynamic program.}
We define $\total(v)$ to be a set of pairs $(\dval, c)$,
where $c$ is the cost of a solution in $P_{vw}$ assuming that at most $|\dval|$ units of demand leave $v$ towards its parent (or arrive if $\dval < 0$),
including opening costs and cable costs for all vertices and edges in $P_{vw}$ and the edge from $v$ to its parent.

The values of the dynamic program for a vertex $v \in P$, $\total(v)$, are computed recursively by considering how the demand can be distributed among its parent and child.

If $v$ is the leaf of the path, we either send the demand of $v$ to the parent,
or we open a facility and allow the parent $p$ to send any amount of demand $\dval$ to $v$ at cost $f(v) + \cost(uv,\dval)$:
\[
\total(w) = \set{(d(w), \cost(pw,d(w)))} \cup \set{(-\dval, f(w) + \cost(pw, \dval)) : \dval \in [0, D]},
\]

To compute $\total(v)$ for an internal node $v$, we use the values of $\total(v')$ for its child node $v'$ and consider whether a facility is opened at $v$.

If a facility is opened at $v$, we consider $\dval \in [-D,0]$, as we can receive any amount of demand from the parent,
and take the minimum cost of $\total(v')$, representing the minimum cost for serving the subpath rooted at $v'$, while allowing to send any demand to $v$.
We obtain
\[
\total_A(v) = \set*{ (-\dval, f(v) + \cost(pv, \dval) + \totalmin(v') ) : \dval \in [0, D] },
\]
with $\totalmin(v') = \min \{ c : (\dval', c) \in \total(v') \}$.

If no facility is opened at $v$, the flow of demand through node $v$ must be conserved, including demand from $v$ itself, as well as demand from or to $p$ and $v'$. Formally, we have
\[
\total_B(v) = \set*{ (\dval, c + \cost(pv, |\dval|)) : (\dval', c) \in \total(v'), \dval = \dval' + d(v) }.
\]
The values of $\total(v)$ are given by taking the non-dominated solutions out of the two sets above, $\total(v) = \nondom(\total_A(v) \cup \total_B(v))$,
where $(\dval,c)$ dominates $(\dval', c')$ if $|\dval'| \geq |\dval|$, $c' \leq c$ and $\dval, \dval'$ have the same sign ($\dval \cdot \dval' \geq 0$).

\subparagraph*{Correctness of the DP.}
Let \(\OPT\) be an optimum solution. For a given vertex \(v\), let \(\dval^*(v))\) be the demand leaving \(v\) towards its parent (if \(\dval^* > 0\)) or arriving at \(v\) from its parent (if \(\dval^* \le 0\)) in the optimum solution \(\OPT\), and \(c^*(v)\) be the total cost of opening facilities and routing demand in \(P_{vw}\) in the optimum solution \(\OPT\), including the cost \(c^*\) associated with the demand \(\dval^*\) arriving/leaving from \(v\) from/to its parent.

The following lemmas imply the correctness of the algorithm.
\begin{lemma}[restate=lempath1,name=*]
    For every \(v \in V\), there is a pair \((\dval^*(v), c) \in \total(v)\) for some cost $c$.
\end{lemma}
\thmToAppendix{lempath1}

\begin{appendixproof}
    We prove this by induction on the path structure from leaf $w$ to any vertex $v$.

    For the leaf $w$, consider the optimal solution.
    \begin{itemize}
	\item If the demand $d(w)$ is sent to the parent, then $\dval^*(w) = d(w)$, and thus we have that $(d(w), \cost(pw,d(w))) \in \total(w)$, as desired.
	\item If a facility is opened, let $\dval^*(w) \le 0$ be the demand arriving from the parent.
    Then, $(\dval^*(w), f(w) + \cost(pw, |\dval^*(w)|)) \in \total(w)$, and thus the base case holds.
    \end{itemize}

    We now show $(\dval^*(v), c_v) \in \total(v)$ for some $c_v$.
    By induction, $(\dval^*(v'), c_{v'}) \in \total(v')$ for child $v'$ of $v$.
    \begin{itemize}
    \item
    If a facility is opened at $v$, $\total(v)$ considers all $\dval$, including $\dval^*(v)$.
    Since by induction $(\dval^*(v'), c_{v'}) \in \total(v')$,
    $\total(v)$ contains $(\dval^*(v), \total_{\min}(v') + f(v) + \cost(pv, \dval^*(v)))$, satisfying the claim.
    \item
    If no facility is opened at $v$, demand flow is conserved, so $\dval^*(v) = \dval^*(v') + d(v)$,
    and $\total(v)$ includes pairs $(\dval, c + \cost(pv, \dval))$ for $(\dval', c) \in \total(v')$ and $\dval = \dval' + d(v)$.
    Since $(\dval^*(v'), c_{v'}) \in \total(v')$, $\total(v)$ contains $(\dval^*(v' ) + d(v), c_{v'} + \cost(pv, \dval^*(v' ) + d(v))) = (\dval^*(v), c_v)$, as required.
    \end{itemize}
    In both scenarios (facility opened or not at $v$), we have shown that $(\dval^*(v), c_v) \in \total(v)$. By induction, this holds for all $v \in V$.
\end{appendixproof}

\begin{lemma}[restate=lempath2,name=*]
    For every $v \in V$, let $(\dval^*(v), c) \in \total(v)$ be the pair corresponding to the optimal solution. Then, $c \leq c^*(v)$.
\end{lemma}
\thmToAppendix{lempath2}
\begin{appendixproof}
    We prove this by induction on the path structure from the leaf $w$ to any vertex $v$.
    For the leaf $w$ with parent $p$, consider the optimal solution.
    \begin{itemize}
        \item If the demand $d(w)$ is sent to the parent, the cost of doing so is $\cost(pw,d(w))$, which matches the dynamic program.
        \item If a facility is opened at $w$, let $\dval^*(w) \leq 0$ be the demand arriving from the parent. Then, the cost of the optimum solution is $f(w) + \cost(pw, |\dval^*(w)|)$, which corresponds to the entry $(-\dval^*(w), f(w) + \cost(pw, |\dval^*(w)|))$ added to $\total(w)$.
    \end{itemize}
    In both cases, the cost of the solution is exactly $c^*(v)$, and thus the base case holds.

    We now show by induction that for $(\dval^*(v), c_v) \in \total(v)$, $c_v \leq c^*(v)$.
    Let $v'$ be the child of $v$ and $p$ the parent of $v$.
    By induction, $(\dval^*(v'), c_{v'}) \in \total(v')$ for the child $v'$ of $v$, with $c_{v'} \leq c^*(v')$.
    \begin{itemize}
        \item If a facility is opened at $v$,
        $c_v
        \leq f(v) + c_{v'} + \cost(uv, |\dval^*(v)|))
        \leq f(v) + c^*(v') + \cost(uv, |\dval^*(v)|))
        \leq c^*(v)$, as the optimum solution, besides paying $c^*(v')$, must also open the facility at $v$ and pay for demand $|\dval^*(v)|$ to arrive from the parent $u$.

        \item If no facility is opened at $v$, demand flow is conserved, so $\dval^*(v) = \dval^*(v') + d(v)$.
        The cost of the optimum solution is $c^*(v') + \cost(pv, |\dval^*(v)|)$, and the cost by the DP is
        $c_v = c_{v'} +  \cost(pv, |\dval^*(v)|) \leq c^*(v') + \cost(pv, |\dval^*(v)|) \leq c^*(v)$, as desired.
    \end{itemize}
\end{appendixproof}

\begin{lemma}[restate=lempath3,name=*]
	The values $\total$ can be calculated in time $O(n\cdot D)$, where $D = \sum_v d_v$.
\end{lemma}
\thmToAppendix{lempath3}

\begin{appendixproof}
The dynamic program computes the values of $\total(v)$ for each vertex $v$ in the path.
For each vertex $v$, $\total(v)$ stores a set of pairs $(\dval, c)$, where $\dval$ represents a demand constraint, and $c$ is the minimum cost of serving the subpath $P_{vw}$ under that constraint.

The possible demand values are $\dval \in [-D,D]$, soat most $O(D)$ values.
This means that the size of $\total(v)$ is $O(D)$.

To compute $\total(v)$, we iterate over all values of $\dval \in [-D, D]$.
Since $v$ lies on a path, it has at most one child, and we compute each entry in $\total(v)$ using the child’s $\total$ table.
For each $\dval$, we consider two cases:

\begin{itemize}
    \item \emph{A facility is opened at $v$}: The entire demand $\dval$ must be served locally at $v$, so the cost is computed directly from the opening cost and the cost of serving $\dval$. To support efficient computation across all $\dval$, we preprocess the child’s $\total$ table into a $\total_{\min}$ array in $O(D)$ time, allowing each relevant cost lookup to be performed in $O(1)$ time.

    \item \emph{No facility is opened at $v$}: We compute the cost by looking up $\total$ of its child and adding the corresponding cost. This also takes $O(1)$ time per entry.
\end{itemize}

Each value of $\dval$ is handled in $O(1)$ time using preprocessed data, so computing $\total(v)$ takes $O(D)$ time. Over all $n$ vertices, the total runtime is $O(n \cdot D)$.
\end{appendixproof}


\subsection{Splittable-Demand \bbfl on Trees}
\label{sec:split:trees}

\appendixSubsection{Splittable-Demand BBFL on Trees}

In this section, we design an algorithm for splittable demand \bbfl on trees.
Let $T$ represent the tree with a root node $r$.
We assume that $T$ is binary by the lemma below.

\begin{lemma}[restate=lemBinary,name=*]
Any tree $T$ can be transformed into a binary tree $T'$ without changing the optimum or the cost of solutions.
\end{lemma}

\thmToAppendix{lemBinary}

\begin{appendixproof}
To transform $T$ into a binary tree $T'$, we perform the following steps:
\begin{itemize}
    \item Node conversion:
        \begin{itemize}
            \item For each node $v$ in $T$ with more than two children, we create a new level of nodes between $v$ and its children as follows (see Figure~\ref{fig:binary_tree_transformation}).
            \item If $v$ has an odd number of children, add a new child node and connect it to $v$.
            \item Until $v$ no longer has children, take two children of $v$, disconnect them from $v$ and make them children of a new vertex.
            \item Once $v$ no longer has children, take all of the newly created vertices that do not currently have a parent, and make them children of $v$. Repeat from the start if $v$ still has more than two children.
        \end{itemize}
    \item Cost assignment:
        \begin{itemize}
            \item We assign an opening cost of infinity to all newly created vertices.
            \item We assign a cost of zero to every edge connecting a newly created vertex to its parent.
        \end{itemize}
\end{itemize}
By performing these transformations, we ensure that the optimal solution for $T'$ directly corresponds to an optimal solution for $T$ and vice-versa.
This is because any optimal solution for $T$ can be mapped to a solution for $T'$ by adding the necessary intermediate nodes without affecting the total cost.
\begin{figure}[t]
	\centering
	\begin{minipage}{0.45\textwidth}
		\centering
		\begin{tikzpicture}[scale=0.7, every node/.style={circle,draw,draw=teal, fill=teal,minimum size=0.1mm}]
			\node (v) at (0,0) {};
			\node (c1) at (-3,-2) {};
			\node (c2) at (-1,-2) {};
			\node (c3) at (1,-2) {};
			\node (c4) at (3,-2) {};

			\draw[thick] (v)--(c1);
			\draw[thick] (v)--(c2);
			\draw[thick] (v)--(c3);
			\draw[thick] (v)--(c4);

		\end{tikzpicture}

	\end{minipage}
	\hfill
	\begin{minipage}{0.45\textwidth}
		\centering
		\begin{tikzpicture}[scale=0.7, every node/.style={circle,draw,draw=teal, fill=teal,minimum size=0.1mm}]
			\node (v) at (0,0) {};
			\node[draw=teal, dashed, fill=white,thick] (x1) at (-2,-2) {};
			\node[draw=teal, dashed, fill=white,thick] (x2) at (2,-2) {};
			\node (c1) at (-3,-4) {};
			\node (c2) at (-1,-4) {};
			\node (c3) at (1,-4) {};
			\node (c4) at (3,-4) {};

			\draw[dashed, teal, thick] (v)--(x1);
			\draw[dashed, teal, thick] (v)--(x2);
			\draw[thick] (x1)--(c1);
			\draw[thick] (x1)--(c2);
			\draw[thick] (x2)--(c3);
			\draw[thick] (x2)--(c4);

			\node[draw=white, fill=white, left, xshift=-6pt] at (x1) {$+\infty$};
			\node[draw=white, fill=white, right, xshift=6pt] at (x2) {$+\infty$};
			\node[draw=white, fill=white, below left, xshift=-20pt,yshift=-5pt] at (v) {$0$};
			\node[draw=white, fill=white, below right, xshift=20pt,  yshift=-5pt] at (v) {$0$};
		\end{tikzpicture}

	\end{minipage}

	\caption[Binary transformation of a tree]{Transformation of an arbitrary tree $T$ into a binary tree $T'$. Newly intermediate nodes have opening cost $+\infty$, and added edges have cost $0$, preserving the cost of solutions.}

	\label{fig:binary_tree_transformation}
\end{figure}
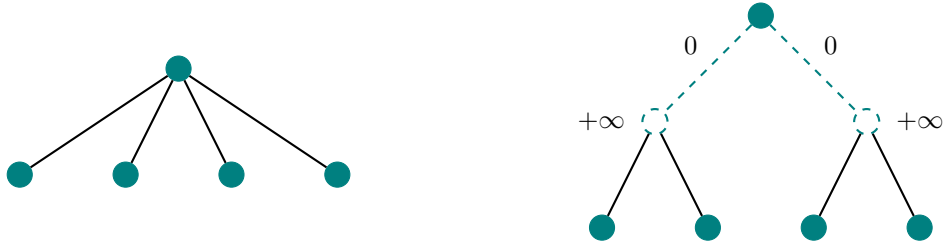
\end{appendixproof}

For each vertex $v \in T$, we define $\total(v)$ as a set of pairs $(\dval, c)$, where $c$ is the minimum cost to serve the subtree $T_v$ with a demand constraint of $\dval$ as follows:
\begin{itemize}
    \item If $\dval>0$, a demand of $\dval$ leaves node $v$ to its parent.
    \item If $\dval\leq 0$, a demand of $|\dval|$ arrives at node $v$ from its parent.
\end{itemize}

In a later step, we will round the costs to multiples of $\epsilon \cdot \opt / n$ and compute the minimum-cost solution using the dynamic program on the modified instance.
As we will show, this results in a $(1+\epsilon)$-approximation to the problem.

\subparagraph{Dynamic program.} The values $\total(v)$ are defined recursively, using the previously computed values for the children.

The base case corresponds to a leaf node $w$ with parent $p$.
We either send the demand to the parent without opening a facility, or we open a facility at $w$.
Formally,
\[
\total(w) = \{ (d(w), \cost(pw, d(w))), (-\dval, f(w) + \cost(pw, \dval)) : \dval \in [0, D] \}
\]

For internal nodes, the DP considers two cases.
If a facility is opened at $v$, we can serve $d(v)$ locally and allow both children ($v'$ and $v''$) and the parent to send demand to $v$. The total cost includes the optimal costs of the subtrees $T_{v'}$ and $T_{v''}$, routing demand from children to $v$, facility opening cost, and routing demand from the parent to $v$.
Formally,
\[
\total_A(v) = \set{
    (-\dval, f(v) +\cost(pv, \dval) + \totalmin(v') + \totalmin(v''))
    : \dval \in [0, D] }
\]
with $\totalmin(v') = \min \{ c : (\dval', c) \in \total(v') \}$.

If no facility is opened at $v$, demand cannot be served there; in this case, all flow at $v$ must pass through it. That means the total demand leaving $v$ toward the parent equals the sum of the flows from both children plus $d(v)$.
The algorithm thus computes the minimum cost for each valid combination of these flows:
\begin{multline*}
\total_B(v) = \set{
    (\dval, \cost(pv, |\dval|) + c' + c'' : \\
    (\dval', c') \in \total(v'), (\dval'', c'') \in \total(v''),
        \dval = \dval' + \dval'' + d(v) }.
\end{multline*}

We then take
$\total(v) = \nondom (\total_A(v) \cup \total_B(v))$,
the non-dominated solutions among those computed.

Let $\operatorname{OPT}$ be an optimum solution and let $\dval^*(v)$ be the demand in the optimum solution, leaving $v$ to its parent or arriving from its parent as above.
Let also $\total^*(v)$ represent the total cost of opening facilities and routing demand in $P_{vw}$, including the demand of value $\dval^*(v)$ arriving/leaving from $v$ from/to its parent.

\begin{lemma}[restate=lemma:tree:dp-correctness-exists, name=*]
\label{lemma:tree:dp_correctness_exists}
    For every \(v \in V\), there is a pair \((\dval^*(v), c) \in \total(v)\) for some cost $c$.
\end{lemma}

\thmToAppendix{lemma:tree:dp-correctness-exists}

\begin{appendixproof}
    We prove this by induction on the height of the subtree $T_v$. For a leaf node $w$ (height 0), in the optimal solution, either the demand $d(w)$ is sent to its parent, resulting in $\dval^*(w) = d(w)$ and $(d(w), \cost(pw, d(w))) \in \total(w)$, or a facility is opened at $w$ with incoming demand $\dval^*(w) \le 0$, leading to $(\dval^*(w), f(w) + \cost(pw, |\dval^*(w)|)) \in \total(w)$. Thus, the base case holds.

    Assume that for children $v'$ and $v''$ of an internal node $v$, there exist pairs $(\dval^*(v'), c') \in \total(v')$ and $(\dval^*(v''), c'') \in \total(v'')$ corresponding to the optimal solution within their subtrees. Consider the optimal solution at $v$. If a facility is opened at $v$, $\total_A(v)$ considers $\dval^*(v)$ and, based on the inductive hypothesis, will contain a cost $f(v) + \totalmin(v')  + \totalmin(v'')  + \cost(pv, \dval^*(v))$, ensuring $(\dval^*(v), c) \in \total(v)$ after pruning. If no facility is opened at $v$, by flow conservation, $\dval^*(v) = \dval^*(v') + \dval^*(v'') + d(v)$. $\total_B(v)$ constructs exactly such combinations of demands and costs from its children, so $(\dval^*(v), c) \in \total(v)$ after pruning.

    In both cases, a pair $(\dval^*(v), c) \in \total(v)$ exists. By induction, this holds for all $v \in V$.
\end{appendixproof}

\begin{lemma}[restate=lemma:tree:dp_correctness_optimal, name=*]
\label{lemma:tree:dp_correctness_optimal}
    For every $v \in V$, let $(\dval^*(v), c) \in \total(v)$ be the pair corresponding to the optimal solution. Then, $c \leq c^*(v)$.
\end{lemma}

\thmToAppendix{lemma:tree:dp_correctness_optimal}

\begin{appendixproof}
    We prove this by induction on the height of the subtree $T_v$. For a leaf node $w$, consider the optimal cost $c^*(w)$. If the optimal solution sends the demand $d(w)$ to the parent, $c^*(w) = \cost(pw, d(w))$, and $\total(w)$ contains $(d(w), \cost(pw, d(w)))$. If the optimal solution opens a facility at $w$ with incoming demand $\dval^*(w)$, $c^*(w) = f(w) + \cost(pw, |\dval^*(w)|)$, and $\total(w)$ contains $(-\dval^*(w), f(w) + \cost(pw, |\dval^*(w)|))$. Thus, the base case holds.

    Assume for all children $v'$ and $v''$ of an internal node $v$, the optimal cost for their subtrees $c^*(v')$ and $c^*(v'')$ is such that $\total(v')$ contains a pair $(\dval^*(v'),c')$ with cost $c'\leq c^*(v')$ in $\dval^*(v')$, and similarly for $v''$. Consider the optimal solution at $v$ with cost $c^*(v)$ and demand $\dval^*(v)$ from the parent.

    If a facility is opened at $v$ in the optimal solution, $c^*(v) = f(v) + c^*(v') + c^*(v'') + \cost(pv, \dval^*(v))$. The DP considers this case in $\total_A(v)$ using $\totalmin(v') \leq c^*(v')$ and $\totalmin(v'') \leq c^*(v'')$, thus finding a cost of at most $c^*(v)$ for $\dval^*(v)$.

    If no facility is opened at $v$ in the optimal solution, $c^*(v) = c^*(v') + c^*(v'') + \cost(pv, \dval^*(v))$ where $\dval^*(v) = \dval^*(v') + \dval^*(v'') + d(v)$. The DP constructs such a cost in $\total_B(v)$ based on the costs of $\total(v')$ and $\total(v'')$ at the corresponding demand values.

    After the $\nondom$ operation, $\total(v)$ will contain a pair $(\dval^*(v), c)$ with $c \leq c^*(v)$. By induction, this holds for all $v \in V$.
\end{appendixproof}

\begin{lemma}[restate=lemma:tree:dp_complexity, name=*]
\label{lemma:tree:dp_complexity}
	The values $\total$ can be calculated in time $O(n\cdot D^2)$, where $D = \sum_v d_v$.
\end{lemma}

\thmToAppendix{lemma:tree:dp_complexity}

\begin{appendixproof}
The dynamic program computes the values of $\total(v)$ for each vertex $v$ in the tree.
For each vertex $v$, $\total(v)$ stores a set of pairs $(\dval, c)$, where $\dval$ represents a demand constraint, and $c$ is the minimum cost to serve the subtree $T_v$ under that constraint.
The possible values for $\dval \in [-D,D]$. Therefore, there are at most $O(D)$ different values for $\dval$. This means that the size of $\total(v)$ is $O(D)$.
To compute $\total(v)$, we need to consider all possible values of $\dval$ and, for each $\dval$, calculate the corresponding cost $c$. In the worst case, calculating the cost for a given $\dval$ could involve examining all possible combinations of how demand can be routed through the children of $v$, which can take time $O(D)$. Thus, computing a single $\total(v)$ takes time $O(D^2)$.
Since we need to compute $\total(v)$ for each of the $n$ vertices in the tree, the total time complexity of the dynamic program is $O(n \cdot D^2)$. Moreover, since converting each tree to a binary tree takes time $O(n)$, this preprocessing step does not change the overall complexity, which remains $O(n \cdot D^2)$.
\end{appendixproof}

The following proof demonstrates the \ptas for Theorem~\ref{thm:intro:bbfl-split:ptas} by detailing the cost-rounding and dynamic programming approach.

\begin{proof}[Proof of Theorem~\ref{thm:intro:bbfl-split:ptas}]
We design a \ptas using a combination of cost rounding and a dynamic programming algorithm on binary trees.
Let $\opt$ be the cost of the optimal solution.
We will do a binary search on a bound $\opt \leq L \leq 2\opt$, starting with $L=\min_v f(v)$, and doubling each time that the algorithm does not find a solution of cost at most $(1+\epsilon)L$.
Notice that, since we start with $L \leq 2\opt$, and the algorithm is sure to stop if $\opt \leq L$, then  $L \leq 2\opt$ for every value of $L$ we consider.

For each value of $L$, we start by rounding the cable installation costs $\ell(e) \cdot \ccap(\dval)$ to the nearest multiple of $\delta=\varepsilon \cdot L / 2n$, for every edge $e$ and demand $\dval$.
Each edge cost changes by at most $\delta$, so the total rounding error across all $n$ vertices is at most $n \cdot \delta = \varepsilon \cdot L/2 \le \varepsilon \cdot \opt$. Hence, rounding introduces at most an additive $\varepsilon \cdot \opt$ error, and thus a solution with rounded costs at most $\opt$ has cost at most $(1 + \varepsilon) \cdot \opt$ using original costs.

We then use a slightly modified version of the dynamic program of Section~\ref{sec:split:trees}, which computes, for each subtree, the minimum cost of satisfying demands under flow constraints.
The algorithm is executed using the rounded-down costs and, throughout the algorithm, we discard any entries that cost more than $L$.
At the end of the algorithm, we return the best solution if it costs at most $L$ (using the rounded costs), or double $L$ and restart.

Lemmas~\ref{lemma:tree:dp_correctness_exists} and~\ref{lemma:tree:dp_correctness_optimal} show that the DP correctly finds the minimum-cost solution under the rounded costs.
Notice that there is always a solution of cost at most $\opt$ using rounded costs, as the costs only decrease.
Therefore, we conclude that if no solution is found, $\opt > L$.
On the other hand, if a solution is found, it must be optimal and thus have rounded cost at most $\opt$, which implies original cost at most $(1+\epsilon)\opt$ as we argued above.

We now remark that Lemma~\ref{lemma:tree:dp_complexity} is true for any $D$ such that, for every $v \in V$, $\total(v)$ contains at most $D$ entries.
Since the costs are rounded to multiples of $\delta$, and we only consider solutions of cost at most $L$, we have at most $D=2n/\epsilon$ possible cost values, and thus the running time is $O(n^3 / \epsilon^2)$.
Thus, the algorithm returns a $(1 + \varepsilon)$-approximate solution in time $O(\epsilon^{-2} \cdot \poly(n))$.
\end{proof}


\section{Cable-Unsplittable-Demand \kcfl on Trees}
\label{sec:unsplit}

This section addresses the cable-unsplittable demand version of \kcfl, where demands on an edge cannot be split over multiple cables, and each cable can only carry demand in a single direction.
Unlike the splittable-demand case handled in the previous sections, we show that cable-unsplittable \kcfl is APX-hard even in very simple instances.
Thus, our goal is to approach the best achievable approximation factor on trees, so as to introduce new techniques to handle this problem.

Our approach is to solve a resource-augmented variant of the problem, where we are allowed to slightly violate the capacities by a factor of $1+\epsilon$.
To clarify, we allow our solution to put total demand $(1+\epsilon)\mu$ on a cable of capacity $\mu$, but importantly, the optimum solution has total demand at most the capacity on each cable it uses.
For more details on resource-augmentation settings, we refer the reader to the book chapter by Roughgarden~\cite{Roughgarden20a}.

We refer to this resource-augmented variant as \rakcfl, and after giving a PTAS for the problem (\Cref{thm:unsplit:ptas-ra}), we show how it can be used to give a $2$-approximation to the original problem (\Cref{lem:unsplit:ra-transform}).
The factor of $2$ is significantly larger than the factor for splittable-demand, but is close to the hardness of $3/2-\epsilon$ given in \Cref{thm:unsplit:apx-hardness}.

We start by stating the APX-hardness result.

\begin{theorem}[restate=thmApxHardness,name=*]
\label{thm:unsplit:apx-hardness}
\kcfl is APX-hard when the input graph is a star and $k=1$.
In particular, it is NP-hard to approximate the problem to a factor of $3/2-\epsilon$ for any $\epsilon > 0$.
\end{theorem}

\thmToAppendix{thmApxHardness}

\begin{appendixproof}
The proof is by a standard reduction from the \emph{partition problem}, and follows a similar proof for \emph{bin packing}~\cite[Sec.~3.3]{WilliamsonS11/0030297/daglib}.
Let $a_1, a_2, \ldots, a_n$ be positive integers with sum $\sum_i a_i = 2b$ for some integer $b$.
The partition problem asks if it is possible to find a set $I \subseteq [n]$ with sum that is half the total,
that is, $\sum_{i \in I} a_i = b$.
This problem is known to be NP-hard, as it belongs to Karp's list of 21 NP-hard problems~\cite{Karp72}.

We reduce an instance of partition to \icfl as follows:
Let $G$ be a star with center $u$ and $n+1$ leaves labeled $w, v_1, v_2, \ldots, v_n$ with:
\begin{itemize}
    \item Edge lengths: $\len(uw) = 1$ and $\len(uv_i) = \epsilon/n$ for every $i \in [n]$;
    \item Demands: $d(v_i) = a_i$, $d(u) = d(w) = 0$.
    \item Facility costs: $f(w) = \epsilon$, $f(u) = f(v_i) = 100$.
    \item Cable type: single cable $(\mu, c) = (b, 1-\epsilon)$.
\end{itemize}

We show that if the partition instance is a YES-instance, then the cost of the optimum solution is at most $2$,
and otherwise the optimum cost is at least $3-\epsilon$.
Thus, no algorithm can distinguish the case of optimum cost at most $2$ from at least $3-\epsilon$, and thus there is no $(3/2-\epsilon)$-approximation algorithm unless P=NP.

For a YES-instance such that $\sum_{i \in I} a_i = b$ for $I \subseteq [n]$, we build a solution as follows:
we open a facility at $w$ with opening cost $\epsilon$;
then transport all of the demands from each $v_i$ to $u$ using 1 cable each for a cost of $\cost(uv_i, a_i) = \epsilon/n \cdot (1-\epsilon )\leq \epsilon/n$ and a total cost over all $n$ demands of at most $\epsilon$;
finally, we use 2 cables to transport demand from $u$ to $w$, one carrying the demand corresponding to $I$, the other to $[n] \setminus I$, for a cost of $2-2\epsilon$;
the total cost is at most $2$.
Notice that the capacity $b$ of the cable is sufficient to carry any individual demand, but also to carry the demands in $I$ or $[n]\setminus I$, as $\sum_{i \in I} a_i = b$ and $\sum_{i \in [n]\setminus I} a_i = 2b - b = b$.

On the other hand, let us consider a NO-instance and assume for contradiction that its cost is less than $3-\epsilon$.
If it has cost less than $3-\epsilon$, then it cannot open a facility on $u$ or any of the $v_i$, as that would incur cost $100$;
thus, it opens a facility on $w$, and all of the demands must be carried to that facility.
W.l.o.g.\ $a_i \leq b$ for every $i\in[n]$, and thus we can transport the demands to $u$ by using a single cable on each edge $uv_i$, for a cost of $\epsilon/n \cdot (1-\epsilon)$ per edge, and in total at most $\epsilon$.
However, as this is a NO-instance, we know that for any set $I \subseteq [n]$, either $\sum_{i \in I} a_i > b$ or $\sum_{i \in [n]\setminus I} a_i > b$,
and thus it is not possible to transport all the demand on 2 cables.
Thus, the solution must use at least $3$ cables, for a total cost of at least $3-3\epsilon + 2\epsilon = 3-\epsilon$, which contradicts the assumption.

This completes the proof of the theorem.
\end{appendixproof}

The main result of this section is a PTAS for \kcfl in the resource augmented setting.

\begin{definition}
Let $(G,\len,f,d,\mathbf{q})$ be an instance of the \kcfl problem, and let $\epsilon > 0$.

A tuple $(I, \mathbf{g}, \mathbf{Q})$ is a \emph{solution with $\epsilon$-resource augmentation} if:
\begin{enumerate}
    \item Each $g_v$ is a flow from $\{v\}$ to $I$ of value $d_v$.
    \item For every edge $uw$, flows $\{g_v(u,w)\}$ and $\{g_v(w,u)\}$ can be partitioned into cables so that:
    \begin{itemize}
        \item Each cable carries total flow at most $(1+\epsilon)\mu$, where $\mu$ is its capacity.
        \item Each flow is assigned to a cable with sufficient capacity.
        \item All flows on the same cable go in the same direction.
    \end{itemize}
\end{enumerate}
\end{definition}

We denote the \kcfl problem in the $\epsilon$-resource augmented setting as $\epsilon$-\rakcfl, for a given $\epsilon > 0$.
Before proving the main result, we show that an approximation for \rakcfl can be leveraged to obtain an approximate solution for the original \kcfl problem.
We remark that the reduction described below applies to general graphs.

\begin{lemma}[restate=lemRaTransform,name=*]
\label{lem:unsplit:ra-transform}
Let $(G,\len,f,d,\mathbf{q})$ be an instance of the \kcfl problem, and $(I, \{g_v\}_v, \{Q_e\}_e)$ be an $\alpha$-approximate solution in the $\epsilon$-resource augmented setting for some $0 < \epsilon < 1/2$ and $\alpha \geq 1$.
Then there exists a $2\alpha$-approximate solution for the \kcfl problem.
\end{lemma}

\thmToAppendix{lemRaTransform}

\begin{appendixproof}
We can modify the $\epsilon$-resource augmented solution to make it feasible for the \kcfl problem as follows:

For each edge $e \in E$ and each cable of capacity $\mu$ carrying a total flow exceeding $\mu$, we add a copy of the same cable to $Q_e$.
This operation increases the cost of the solution by at most a factor of $2$.

We then partition the flows assigned to the original cable between the two copies: sort the flows in decreasing order of value, assign as many as possible to the first copy without exceeding its capacity, and place the remaining flows on the second copy.
Since the first copy receives at least half of the total flow, the second copy carries at most $\mu$, ensuring feasibility.

After applying this procedure to all edges, we obtain a solution that respects the original capacities and has cost at most twice that of the $\alpha$-approximate $\epsilon$-resource augmented solution.
Hence, it is a $2\alpha$-approximate solution for \kcfl.
\end{appendixproof}
The main result of this section can be restated in a slightly different form as follows, equivalent to \Cref{thm:intro:kcfl-unsplit:ra-ptas}:

\begin{theorem}
\label{thm:unsplit:ptas-ra}
Let $(G, \len, f, d, \mathbf{q})$ be an instance of \kcfl on a tree, and let $0 < \epsilon \le 1/10$.

There exists an algorithm for $\epsilon$-\rakcfl that computes an $\epsilon$-resource augmented solution of optimal cost in time $n^{k \cdot f(\epsilon)}$, for some computable function $f$.
\end{theorem}

\begin{proof}[Proof of \Cref{thm:intro:kcfl-unsplit:apx}]
We first compute an $\epsilon$-resource augmented optimal solution for $\epsilon$-\rakcfl on the given tree instance using \Cref{thm:unsplit:ptas-ra}.
By applying the transformation described in \Cref{lem:unsplit:ra-transform}, this solution can be converted into a feasible solution for the original \kcfl problem.
Choosing $\epsilon = 1/10$ ensures that the resulting solution achieves a $2$-approximation of the optimal cost.
\end{proof}

In the rest, we focus on proving \Cref{thm:unsplit:ptas-ra}.
For simplicity and clarity, we first consider instances with a single cable type in Section~\ref{sec:unsplit:k1}, and then in Section~\ref{sec:unsplit:k} we generalize the results to a constant number of cable types.

\subsection{\rakcfl on Trees for a Single Cable Type (k=1)}
\label{sec:unsplit:k1}
\appendixSubsection{RAkCFL on Trees for a Single Cable Type (k=1)}

Our approach relies on two main ideas: first, rounding demands so that similar values can be treated as identical, reducing the number of distinct values; second, imposing structure on the solution by grouping small demands into bundles along the same paths, limiting the combinatorial explosion caused by many small demands. While the first idea contributes to underestimating the value of demands, the second causes demands to be served along sub-optimal paths, and thus both lead to resource augmentation.

Let $(\mu, c)$ denote the single cable type.
For each demand of at least $\epsilon \mu/2$, we round it down to the nearest multiple of $\epsilon^2 \mu$;
as we will show, these demands shrink by at most $\epsilon^2 \mu$, and thus contribute to an excess of at most $2\epsilon \mu$ in the demand of each cable.
Small demands must be handled differently:
they are grouped into \emph{bundles} totaling at least $\epsilon \mu/2$, with an additional leftover bundle of total demand less than $\epsilon \mu / 2$.
Both small demands and leftover bundles are rounded more finely, to multiples of $\epsilon^2 \mu / n$, and become non-leftover bundles once their combined demand reaches at least $\epsilon \mu/2$.
The demands in each bundle cannot be separated once joined and must be served along the same path by the same facility.
We later show that any solution can be modified to follow this structure, at the price of $2\epsilon$-resource augmentation.

In order to use dynamic programming,
we must then show that any solution to the original problem can be modified so that small demands are grouped, which overloads cables by a small fraction of the capacity.
This, together with the increased demand of the original values instead of the rounding, leads to an excess of $\emu$ on each cable,
but allows us to write a dynamic programming to obtain a solution with optimum cost.

The dynamic program works by considering, for each vertex $v$, tuples that specify the number of large demands and bundles of each rounded size, the demand in the leftover bundle,
both going up to the parent and down from the parent, as well as the minimum cost to serve these demands in the subtree $G_v$.
The demand on each cable is allowed to be at most $(1+2\epsilon)\mu$.

Solving the dynamic program yields a $2\epsilon$-resource augmented optimal solution,
and converting the solution  back to the original demands incurs a factor of $(1+2\epsilon)$, for a total demand of at most $(1+5\epsilon)\mu$ on each cable.

Most of the technical work in this section is to ensure that rounding and bundling preserve sufficient demand to guarantee an optimal-cost, resource-augmented solution.
For simplicity, we present the algorithm achieving $5\epsilon$-resource augmentation; the theorem follows by rescaling $\epsilon$.

\subsubsection{Rounding demands.}
We begin by describing the demand rounding procedure.
Assume that $\ieps$ is an even integer, and define the rounding units
\[
\delta_b = \epsilon^2 \mu, \qquad \delta_s = \frac{\epsilon^2 \mu}{n},
\]
for large and small demands, respectively.

For each vertex $v \in V$, we set
\[
(b_v, s_v) :=
\begin{cases}
\bigl(\lfloor d_v / \delta_b \rfloor, 0\bigr) & \text{if } d_v \geq \epsilon \mu / 2,\\[2mm]
\bigl(0, \lfloor d_v / \delta_s \rfloor\bigr) & \text{if } d_v < \epsilon \mu / 2.
\end{cases}
\]
Observe that $b_v \in \{0\} \cup [\ieps/2, \ieps[2]]$ and $s_v \in [0, \ieps n/2-1]$, and that exactly one of $b_v$ or $s_v$ is non-zero for each vertex $v$.

Two observations will be useful when comparing the solution computed on rounded demands to the original demands:

\begin{observation}\leavevmode
\label{obs:unsplit:rounding}
Let $(b_v, s_v)$ be the rounded values defined above. Then:
\begin{enumerate}
\item If $d_v \geq \epsilon \mu / 2$, then $b_v \delta_b$ is a $(1+2\epsilon)$-approximation of $d_v$.
\item For any subset $S \subseteq V$ of small demands, the total of the original demands and the total of the rounded demands differ by at most $\epsilon^2 \mu$.
\end{enumerate}
\end{observation}

\begin{proof}
For the first point, note that
\[
b_v \delta_b \leq d_v \leq (b_v+1)\delta_b \leq b_v \delta_b (1 + 2\epsilon).
\]

For the second point, for any subset $S \subseteq V$ of small demands, we have
\[
\sum_{v \in S} s_v \delta_s \leq \sum_{v \in S} d_v \leq \sum_{v \in S} (s_v + 1) \delta_s
= \sum_{v \in S} s_v \delta_s + |S| \delta_s \leq \sum_{v \in S} s_v \delta_s + n \delta_s = \sum_{v \in S} s_v \delta_s + \epsilon^2 \mu.
\]
\end{proof}

\subsubsection{Dynamic program.}
The dynamic program will compute, for every possible set of demands rounded as indicated above, the minimum cost needed to serve the demands in the subtree $G_v$, by opening facilities in the subtree and paying for cable installation in the edges of $G_v$ and the edge from $v$ to its parent, given the demands going to or coming from the parent.

Formally, for each vertex $v \in V$ with parent $u \in V$, we compute a set of non-dominated solutions $\total(v)$ consisting of tuples of the form $(a, \dval_s, \bar{\dval}_s, c)$, where
\[
\dval_s, \bar{\dval}_s \in [0, \ieps n/2-1], \quad
a \in \aset := \set{0}^{\ieps/2} \times [-n,n]^{\ieps[2]-\ieps/2},
\]
i.e., $a$ has $\ieps[2]$ values between $-n$ and $n$, but the first $\ieps/2$ are fixed to $0$ for convenience,
and $\aset$ is the set of possible values for $a$.
We also define the non-negative subset of $\aset$,
\[
\asetplus := \set{0}^{\ieps/2} \times [0,n]^{\ieps[2]-\ieps/2}.
\]

A tuple $(a, \dval_s, \bar{\dval}_s, c)$ represents a solution where:
\begin{itemize}
    \item $a_i \in [-n,n]$ is the number of large demands of value $\delta_b i$ with $i \in [\ieps/2, \ieps[2]]$ traveling up to $u$ if $a_i > 0$, or coming down from $u$ if $a_i < 0$,
    \item $\delta_s \dval_s$ (resp.\ $\delta_s \bar{\dval}_s$) is the total of small demands traveling up to (resp.\ coming down from) $u$,
    \item $c$ is the cost of serving all demands in $G_v$, plus the demands coming from the parent as specified by $(a, \bar{\dval}_s)$, except the demands going to the parent as specified by $(a, \dval_s)$; this cost includes facility opening and cable installation on all edges of $G_v$ and the edge $uv$.
\end{itemize}

One solution dominates another if and only if the values of $a, \dval_s, \bar{\dval}_s$ are identical and the cost of the first solution is lower.

We also note that it suffices to consider demands of value $\delta_b i$ in a single direction for each edge: if an edge carries demands in both directions, they can be ``uncrossed'' without increasing cost (\Cref{lem:uncrossing:edges:unsplit}).

    \begin{lemma}
    \label{lem:uncrossing:edges:unsplit}
    In the context of unsplittable or cable-unsplittable demands and monotone installation costs, there is an optimal solution where demands with the same value do not cross;
    that is, for any edge $uw$, two demands of value $\dval$ cannot pass one in the direction $(u,w)$ and the other in direction $(w,u)$ simultaneously.
    \end{lemma}

    \begin{proof}
    Assume by contradiction no such optimal solution exists, and consider the optimal solution with the fewest crossings.
    Similarly to the proof of \Cref{lem:uncrossing:edges}, we take an edge $uv$ that has a crossing and two demands of value $\dval$ that cross at that edge.
    We can switch the paths at $u$ and $v$, so the first demand follows its path up to $u$, then the path of the other demand after $u$, while the second follows its path up to $v$, and then the first after $v$.
    This preserves feasibility of the solution, since both demands have the same value, and does not increase the cost, as the demand on uv decreases.
    This solution is thus optimal and has one fewer crossings, which contradicts the assumption.
    \end{proof}

The set $\total(v)$ of tuples $(a, \dval_s, \bar{\dval}_s, c)$ can be computed recursively starting from the leaves and moving up in the tree as follows.
Recall that we use $\indic(b_v)$ to refer to the vector that is all zeros except a $1$ at position $b_v$;
we consider $\indic(0)=\zeroes$ for simplicity.
We write $\cost(uv, a, \dval_s)$ to be the minimum cost of installing cables for the demands specified by $\max(a,\zeroes) \in \asetplus$ and $\dval_s \in [0,\ieps n/2-1]$
so that each cable gets total demand at most $(1+2\epsilon)\mu$;
we later show how this value can be computed in polynomial time.
The cost of an edge is then represented as $\cost(uv, a, \dval_s) + \cost(uv, -a, \bar\dval_s)$.

\subparagraph*{Leaf nodes.} Let \(v\) be a leaf with parent \(u\). We distinguish two cases.
\begin{description}
    \item[No facility is opened at \(v\).]
The demand of \(v\) is routed toward its parent, so we add the tuple
\(
    \bigl(\indic(b_v),s_v,0,
    \cost(uv,\indic(b_v),s_v)\bigr)
\)
to \(\total(v)\).

\item[A facility is opened at \(v\).]
For every \(a\in\asetplus\) and
\(\bar\dval_s\in[0,\ieps n/2-1]\), we add the tuple
\(
    \bigl(-a,0,\bar\dval_s,
    f(v)+\cost(uv,a,\bar\dval_s)\bigr)
\)
to \(\total(v)\).
\end{description}

\subparagraph*{Internal nodes.} Now let \(v\) be an internal node with parent \(u\) and children \(v'\) and \(v''\). We again distinguish two cases.

\begin{description}
\item[No facility is opened at \(v\).]
Fix two tuples
    $(a',\dval'_s,\bar\dval'_s,c')\in\total(v')$
    and
    $(a'',\dval''_s,\bar\dval''_s,c'')\in\total(v'')$.

We describe how to combine these tuples when no facility is opened at \(v\). The large-demand vectors satisfy a flow-balance constraint. The only additional complication is that leftover bundles of small demands may be merged at \(v\); whenever their combined size reaches \(\emu/2\), they are promoted to a large demand.

Let $B_v:=\set{v',v'',u}$ and $D_v:=B_v\cup\set{v}$.
For a candidate value \(\bar\dval_s\in[0,\ieps n/2-1]\), define the small-demand amounts arriving at \(v\) by
    $\alpha_{v'}:=\dval'_s$,
    $\alpha_{v''}:=\dval''_s$,
    $\alpha_v:=s_v$,
    $\alpha_u:=\bar\dval_s$.

We enumerate every assignment
\(
    h\colon D_v\to B_v
\)
such that \(h(x)\neq x\) for every \(x\in B_v\). Thus, \(h\) specifies the branch along which each incoming leftover bundle is routed, and prevents a bundle from immediately returning along the edge from which it arrived.
For each \(y\in B_v\), let
\(
    T_y:=\sum_{x\in h^{-1}(y)}\alpha_x
\)
be the total amount assigned to branch \(y\). We determine a promoted large-demand index \(b_y\) and a residual small-demand amount \(r_y\) as follows:
\[
(b_y,r_y):=
\begin{cases}
    (0,T_y),
        & \text{if }T_y\delta_s<\emu/2,\\[2mm]
    \left(
        \ffloor{T_y\delta_s/\delta_b},0
    \right),
        & \text{if }\emu/2\leq T_y\delta_s<\emu,\\[2mm]
    \left(
        \ffloor{(T_y-\alpha_{x_y})\delta_s/\delta_b},
        \alpha_{x_y}
    \right),
        & \text{if }\emu\leq T_y\delta_s,
\end{cases}
\]
where, in the last case, we enumerate every possible choice of \(x_y\in h^{-1}(y)\). In other words, all but one of the bundles
assigned to \(y\) are promoted to a large demand, while the remaining bundle stays as the leftover bundle.
We retain the combination only if its residual bundles agree with the two child states, namely,
\[
    r_{v'}=\bar\dval'_s
    \qquad\text{and}\qquad
    r_{v''}=\bar\dval''_s.
\]
We then set \(\dval_s:=r_u\) and define
\[
    a:=a'+a''+\indic(b_v)
       +\sum_{y\in B_v}\indic(b_y).
\]
Finally, we add the tuple
\(
    (a,\dval_s,\bar\dval_s,c)
\)
to \(\total(v)\), where
\[
    c:=c'+c''
       +\cost(uv,a,\dval_s)
       +\cost(uv,-a,\bar\dval_s).
\]
\item[A facility is opened at \(v\).]
In this case, no demand is routed from \(v\) toward either child.
Therefore, for every pair of tuples
    $(a',\dval'_s,\bar\dval'_s,c')\in\total(v')$
    and
    $(a'',\dval''_s,\bar\dval''_s,c'')\in\total(v'')$,
with \(a',a''\in\asetplus\), and every
\(a\in\asetplus\) and
\(\bar\dval_s\in[0,\ieps n/2-1]\), we add to $total(v)$ the tuple
\[
    (-a,0,\bar\dval_s,c), \qquad \text{where } c:=f(v)+c'+c''+\cost(uv,a,\bar\dval_s).
\]
\end{description}

After considering both cases, we retain only the non-dominated tuples
in \(\total(v)\).

\subsubsection{Computing costs.}
We now specify how to compute $\cost(uv, a, \dval_s)$ for a single cable $(\mu,c)$.
As it corresponds to the minimum cost of installing demands $\max(a,\zeroes) \in \asetplus$, we assume w.l.o.g.\ that $a \in \asetplus$; we also assume that $\len(uv)=1$.
We remark that the values for $\cost(uv, a, \dval_s)$, $\len(uv)=1$ can be pre-computed for any $a \in \asetplus$,
and then simply adapted to $\dval_s$ and the length of the edge when needed in the algorithm.

We compute the costs for every $a \in \asetplus$ simultaneously by using dynamic program, $C[a]$ refers to the smallest bag to serve demands $a$ so each cable is assigned demand $(1+2\epsilon)\mu$.
Let $\aset' = \set{a' \in \asetplus : \sum_i i\cdot a'_i \leq \ieps[2](1+2\epsilon)}$ be the set of possible configurations for a cable.

$C$ is computed as follows: $C[\zeroes]=0$, $C[a] = \min\set[\big]{C[a-a']\cup \set{a'}: a' \in \aset', a-a' \in \asetplus}$, where $\min$ selects a smallest set.
During the algorithm, the cost of the cables is given $C[a]\cdot c$, and the leftover bundle given by $\dval_s$ is added to the cable with least demand, making use of resource augmentation.
If no cable exists (because $a=\zeroes$) but $\dval_s > 0$, use 1 cable.

\subsubsection{Analysis of the algorithm.}

Let $(I, \mathbf g, \mathbf Q)$ be the solution obtained from running the DP and recovering the best solution.

\begin{claim}[restate=claim:unsplit:1cfl:solution,name=*]
\label{claim:unsplit:1cfl:solution}
Let $(I^*,\mathbf{g^*}, \mathbf{Q^*})$ be a solution.
Then there is a modified solution $(I^*,\mathbf{g'}, \mathbf{Q^*})$ that opens the same facilities, uses the same cables for each edge, and routes large demands along the same paths, but groups small demands as specified in the dynamic program, so that the total rounded-down demand on each cable is at most $(1+2\epsilon)\mu$.
\end{claim}

\thmToAppendix{claim:unsplit:1cfl:solution}

\begin{appendixproof}
\newcommand{\gdem}[3]{\ensuremath{\gamma_{#1}(#2,#3)}}

We will modify the given solution $(I^*,\mathbf{g^*}, \mathbf{Q^*})$ such that demands are grouped in a way that is compatible with the dynamic program, and the excess demand on each edge is less than $2\emu$ after rounding.

We will process the solution twice:
once from the bottom up to rearrange the demands from the children onto the parent,
and then a second one from the top down to rearrange the demands from the parent to the children.
We will change the solution at each vertex so that small demands are grouped further in a way that is consistent with the grouping at its children and parent.

We do not make any changes to large demands, only to bundles of small demands.
For this process to work, we require that the rearranged demands are less than $\emu$.
For this reason, it is important that small demands come in bundles of less than $\emu$, including the leftover bundle that totals less than $\emu/2$.
We consider only the small demand $s_v$ associated to each vertex.
We also assume that there is no facility at $v$, as otherwise all of the demand incoming to $v$ is served there and no demand goes out.

Let $v$ be a vertex with children $v', v''$ and parent $u$.
We define $\gdem{}{x}{y}$ to be the total amount of small demands served by a path that passes through $x$ and afterwards $y$, where $x, y \in \set{v', v'', u, v}$.
We will route small demands between vertices $v'$, $v''$, $u$ in such a way that the total demand on the edges $vv'$, $vv''$, $uv$ does not increase, with the exception of at most two extra bundles of demands, which are always routed from $v$ to its children $v'$ and $v''$.
These additional demands are then routed when considering $v'$ and $v''$ in the second stage, when routing demands from $v$, their parent.
As a reminder, there can be any number of bundles of demands, but there is only one bundle of ungrouped demand, of total value at most $\emu/2$.

At a first stage, we route demands from $v'$ and $v''$ as follows.
For any leaf $v$, we simply send its small demand $s_v$ to the parent $u$ as an ungrouped demand.
For an internal vertex $v$, we start by considering the bundles coming from $v'$:
while there are at least 2 (with total demand $\dval$ and $\dval' \leq \dval$), we can send the smallest towards either $u$ if $\dval' \leq \gdem{}{v'}{u}$ or $v''$ if $\dval' \leq \gdem{}{v'}{v''}$.
This must always be possible, since the demand of the small bundles going from $v'$ to $v$ is, by induction, at most the demand going from $v'$ to $v$, which is $\gdem{}{v'}{u}+\gdem{}{v'}{v''}$.
We then remove the routed bundle from consideration and subtract $\dval'$ from $\gdem{}{v'}{x}$ where $vx$ was the edge chosen to route the bundle.
This process is repeated until there is only one bundle left, and then we apply the same reasoning to send the ungrouped small demands either to $u$ or $v''$.
We then do the same process for $v''$, routing all but one of the bundles of small demands.

All that is left to do at this stage is to create a new bundle of small demands, if the amount of ungrouped small demand going to $u$ is at least $\emu/2$, and to route the last bundle of $v'$ and $v''$.
For the small demands,
we make a bundle with the ungrouped small demands routed to $u$ if the total is at least $\emu/2$ and less than $\emu$;
if the total is at least \emu,
we make a bundle using only two of them (say the small demands of $v'$ and $v''$), and keep the third as ungrouped.
For the last small bundles of $v'$ and $v''$, of demand $\dval'$, $\dval''$,
we route one of them to $u$ if that is still possible, considering the total demand into $u$ (routed earlier from $v'$ and $v''$); the remaining bundle(s) get routed to the respective sibling as an extra demand.
In other words, if $\dval' \leq \gdem{}{v'}{u}+\gdem{}{v''}{u}$ (w.l.o.g. $\dval'\leq \dval''$), we route the bundle of $v'$ to $u$, and add the bundle of $v''$ as an extra demand on $v'$;
otherwise we add the bundle of $v'$ as an extra demand of $v''$ and the bundle of $v''$ as an extra demand of $v'$.

Once the first stage has finished, we now have grouped the small demands going up the tree, and have not so far increased the demands on each edge, other than the single extra demand added to some edges $vv'$ and $vv''$.

For the second stage, we go through each vertex $v$ starting at the root of the tree and going down, with the purpose of routing the demands coming to $v$ from its parent, including two extra bundles of small demands
(one from the first stage, and one more which we may add in the second stage).
For the root, as there is no parent, there are no demands to route, and thus we are trivially done.
For any other vertex $v$ with parent $u$ and children $v'$, $v''$, we repeat the same process to route the small bundles of demands coming from $u$ to $v$ as we did for $v'$ and $v''$:
if there are at least two bundles, we route the smallest of them to either $v'$ or $v''$, remove it from consideration, and decrease the corresponding value $\gdem{}{u}{x}$, routing the ungrouped small demands in the same way.

As to the last bundle, we argue that either we can route it to $v'$ or $v''$, or it must be that in the first stage we routed the last bundle of either $v'$ or $v''$ towards $u$, and thus one of $v'$, $v''$ does not yet have an extra demand.
Indeed, if we consider the last bundles of each of $u$, $v'$ and $v''$, with value $\tilde \dval$, $\dval'$, $\dval''$, the sum of their demands must be at most the total demand going into $v$ in the solution, which equals the demand out of $v$,
that is
\[
    \tilde \dval + \dval' + \dval''
    \leq (\gdem{}{v'}{u} + \gdem{}{v''}{v'})
        + (\gdem{}{u}{v'} + \gdem{}{v''}{u})
        + (\gdem{}{v'}{v''} + \gdem{}{u}{v''}),
\]
which by an averaging argument implies that one of the last bundles can be routed to one of the other vertices.
The two remaining last bundles get routed to $v'$ and $v''$ as an extra demand each.
Thus, if the last bundle of $u$ is routed, $v'$ and $v''$ get an extra demand from $v''$ and $v'$, respectively;
if the last bundle of $v'$ is routed, the last bundle of $u$ is routed to $v''$ and the last bundle of $v''$ to $v'$; and analogously for $v''$.
The (at most) two extra demands sent from $u$ to $v$ are routed one each to $v'$ and $v''$, and so each of them gets at most two extra demands.

So far, we have argued that we can transform any solution so that small demands are grouped as considered in our dynamic program, and the capacity on each edge is exceeded by at most two bundles, which increases the demand by at most $2\emu$.
However, it is also necessary to argue that the same cables with resource augmentation can handle its assigned demands.

For any edge $(v,w)$ (in a single direction), consider the set of cables used by the given solution, and assign the large demands to the cables in the same way as in the solution;
we then assign the bundles of small demands greedily to the cable that has the most leftover capacity.

As we know that, with the exception of the two extra demands, the sum of demands on $(v,w)$ in the constructed solution is at most the total capacity of the cables in the given solution,
then it must be the case that when adding any bundle, there is still leftover capacity in one of the cables.
Thus, after adding every non-extra bundle of demands, the demand of each cable cannot exceed the capacity by more than $\emu$, the maximum size of a bundle.
Furthermore, there must be a cable where the demand does not exceed the capacity, as otherwise the total demand would exceed the total capacity.
The two extra demands are placed on one such cable.
Overall, the capacity of each cable is exceeded by at most $2\emu$, as desired.
\end{appendixproof}

\begin{claim}[restate=claim:unsplit:1cfl:feasible,name=*]
\label{claim:unsplit:1cfl:feasible}
$(I, \mathbf g, \mathbf Q)$ is a $5\epsilon$-resource augmented solution.
\end{claim}

\thmToAppendix{claim:unsplit:1cfl:feasible}

\begin{appendixproof}
The solution is feasible by construction, as demands get served only at facilities and are not separated into multiple cables, due to the equations in the recursion, which only have one variable on the right hand side, and the construction of the cables, which takes demands as indivisible units.

The cost of a solution is also correctly computed: the cost of facilities is considered, and the cost of cables takes into account all of the demands going through the edge.

Finally, to see that the solution is $5\epsilon$-resource augmented, we argue that each cable is overloaded by $2\epsilon\mu$ by construction, and then by $2\epsilon\mu$ by the rounding of large demands.

Let a cable be assigned the demands given by a vector $a \in \asetplus$ with small demands $\lambda_s \in [0,\ieps n/2-1]$.
The rounded demand on the cable can be written as
\[
    \sum_i a_i \cdot i\delta_b + \lambda_s \delta_s = \sum_j \tilde d_j + \lambda_s \delta_s,
\]
where $\tilde d$ is a vector of demands corresponding to $a$, containing $a_i$
elements of value $i\delta_b$, for every $i \in [\ieps[2]]$.

Let $U \subseteq V$ be the set of vertices whose demands are carried by the cable, and let $U = U_s \uplus U_b$ be the partition of $U$ into small and large demands, respectively.

Of these values $\tilde d_j$, some correspond to the grouped small demands of a set $S_j \subseteq U_s$, and so we can write:
\[
\tilde d_j
\leq \sum_{v \in S_j} s_v \delta_s
\leq \tilde d_j + \epsilon^2\mu
\leq \tilde d_j (1+2\epsilon),
\]
where the first two inequalities come from the fact that we rounded down the total small demand of the bundle to the nearest multiple of $\delta_b = \epsilon^2\mu$, and the last from the fact that $\tilde d_j \geq \epsilon \mu/2$ and \Cref{obs:unsplit:rounding}.

Similarly, for the demands $\tilde d_j$ corresponding to the large demand of a vertex $u$, we have that $\tilde d_j = b_v \delta_b$ and thus
\(
\tilde d_j
\leq d_v
\leq \tilde d_j + \epsilon^2\mu
\leq \tilde d_j (1+2\epsilon),
\)
and for the small demands $s_v\delta_s \leq d_v \leq s_v\delta_s + \epsilon^2\mu/n$, as a consequence of rounding down to multiples of $\delta_s = \epsilon^2 \mu / n$.

Combining all of these facts, we get that
\begin{align*}
\sum_{v \in U} d_v
&= \sum_{v \in U_s} d_v + \sum_{v \in U_b} d_v \\
&\geq \sum_{v \in U_s} s_v\delta_s + \sum_{v \in U_b} b_v\delta_b \\
&\geq \sum_j \tilde d_j + \lambda_s \delta_s,
\end{align*}
since every demand $\tilde d_j$ corresponds to either a large demand or a bundle of small demands, the remaining of which are accounted for in $\lambda_s \delta_s$.

On the other hand, since $\sum_j \tilde d_j \leq \mu (1+2\epsilon)$, we get
\begin{align*}
\sum_{v \in U} d_v
&\leq \sum_{v \in U_s} \paren*{s_v\delta_s+\frac{\epsilon^2\mu}{n}} + \sum_{v \in U_b} b_v\delta_b(1+2\epsilon) \\
&\leq \sum_j \tilde d_j(1+2\epsilon) + \lambda_s \delta_s + n \frac{\epsilon^2\mu}{n} \\
&\leq \mu(1+2\epsilon)(1+2\epsilon) + \frac{\epsilon \mu}{2} + \epsilon^2 \mu \\
&\leq \mu(1+5\epsilon),
\end{align*}
where again we use that each $\tilde d_j$ corresponds to either a large demand or bundle of small demands, that the remaining small demands total at most $\epsilon\mu/2$, and that $\epsilon \leq 1/10$.
\end{appendixproof}

\begin{claim}[restate=claim:unsplit:1cfl:optimal,name=*]
\label{claim:unsplit:1cfl:optimal}
$(I, \mathbf g, \mathbf Q)$ has optimal cost for the \icfl instance,
i.e.\ its cost is at most the cost of an optimum solution not using resource augmentation.
\end{claim}

\thmToAppendix{claim:unsplit:1cfl:optimal}

\begin{appendixproof}
By \Cref{claim:unsplit:1cfl:solution}, there is a solution with the required properties and optimal cost that uses the rounded-down demands and has total demand on each cable at most $(1+2\epsilon)\mu$.
Thus, the proof follows by optimal substructure, as any solutions to a subproblem that have the same demand profile are functionally equivalent.
\end{appendixproof}

\begin{claim}[restate=claim:unsplit:1cfl:time,name=*]
\label{claim:unsplit:1cfl:time}
The algorithm to compute $(I, \mathbf g, \mathbf Q)$ runs in time $n^{O(\ieps[2])}$.
\end{claim}

\thmToAppendix{claim:unsplit:1cfl:time}

\begin{appendixproof}
We remark that $\aset$ has size $M = (2n)^{\ieps[2]} = n^{O(\ieps[2])}$.
The size of the dynamic program, the running time of trying all possible combinations for the recursive rules, and the time to compute the installation costs for a given configuration are all polynomial in $M$ and $n$, and thus the running time is $n^{O(\ieps[2])}$.
\end{appendixproof}

\begin{proof}[Proof of \Cref{thm:unsplit:ptas-ra} for a Single Cable Type (k=1)]
Consider the solution $(I, \mathbf g, \mathbf Q)$ produced by the dynamic program for the single cable type case.
By \Cref{claim:unsplit:1cfl:solution}, any feasible solution can be transformed so that small demands are grouped according to the DP without increasing the total demand on any cable by more than $2\epsilon \mu$.
\Cref{claim:unsplit:1cfl:feasible} then guarantees that $(I, \mathbf g, \mathbf Q)$ is feasible and constitutes a $5\epsilon$-resource augmented solution.
\Cref{claim:unsplit:1cfl:optimal} ensures that the DP computes an optimal-cost solution among all solutions with the same demand grouping, which implies optimality under resource augmentation.
Finally, \Cref{claim:unsplit:1cfl:time} establishes that the DP runs in time $n^{O(\ieps[2])}$.

Hence, for $k=1$, the dynamic program yields an $\epsilon$-resource augmented solution of optimal cost within the claimed running time.
\end{proof}

\subsection{\rakcfl on Trees for a Constant Number of Cables}
\label{sec:unsplit:k}
\appendixSubsection{RAkCFL on Trees for a Constant Number of Cables}

In this section, we generalize the results of \Cref{sec:unsplit:k1} to a constant number of cables.
We will emphasize the differences with the case of $k=1$ both in the algorithm and the analysis , which are mostly details on how to store the necessary information about the solutions and how to compute the cost.

Let $(\mu_1, c_1), (\mu_2, c_2), \ldots, (\mu_k, c_k)$ be the cables in increasing order of capacity
(if two cables have the same capacity, we take one with lowest cost and discard the others).
We now consider large demands to be demands that are not too small compared to the smallest cable, that is, at least $\epsilon\mu_1/2$.
Demands that are below that threshold are considered small, and rounded down to the nearest multiple of $\delta_s = \epsilon^2\mu_1/n$.

\newcommand{\bvi}[2]{\ensuremath{b^{(#2)}_{#1}}}

Large demands are rounded down to the nearest multiple of $\set{\delta_1, \ldots, \delta_k}$ where $\delta_i = \epsilon^2\mu_i$,
that is, when rounding $d_v$, we consider the quantities $\delta_1\ffloor{d_v/\delta_1}$, $\delta_2\ffloor{d_v/\delta_2}$, \ldots, $\delta_k\ffloor{d_v/\delta_k}$ and round $d_v$ to the closest of these quantities (which is also the largest).
This method has the advantage of providing an approximation that is within $\delta_i$ for any cable $(\mu_i, c_i)$ thus allowing its use in any cable.
The number of possibilities for small demands remains $\ieps n/2-1$.

Formally, we take $s_v = \ffloor{d_v/\delta_s}$  if $d_v$ is small (at most $\emu/2$) and $s_v = 0$ otherwise,
and define $\bvi{v}{i} = \ffloor{d_v/\delta_i}$ if $\delta_i\ffloor{d_v/\delta_i}$ is the closest approximation for $d_v$ such that $d_v \geq \epsilon\mu_i/2$ and $d_v \leq \mu_i$, with $\bvi{v}{j} = 0$ for $j \neq i$.
Note that, for each vertex $v \in V$, only (at most) one of the values among $s_v$ and $\bvi{v}{i}$ is non-zero, by definition.
Additionally, the number of rounded-down possibilities for large demands is now at most $k \ieps[2]$, as for each cable there are $\ieps[2]$ possible values up to its capacity;

\subsubsection{Dynamic program.}
A set of demands is again represented by a vector $a \in \aset$ and two small demands $\dval_s, \bar\dval_s \in [0,\ieps n/2-1]$, though the definition changes to
$\aset := \paren[\big]{\set{0}^{\ieps/2} \times [-n,n]^{\ieps[2]-\ieps/2}}^k$,
accounting for the $k$ cables.
As before, we also define $\asetplus$ to be the subset of $\aset$ where every value is non-negative.

The set $\total(v)$ is computed recursively following a similar structure.
In this section, we use $\indic(\bvi{v}{\cdot})$ to refer to the concatenation of the indicator vectors $\bvi{v}{1},\bvi{v}{2},\ldots,\bvi{v}{k}$.

If $v \in V$ is a leaf, we add to $\total(v)$:
\begin{itemize}
    \item the tuple $(\indic(\bvi{v}{\cdot}), s_v, 0, \cost(uv, \indic(\bvi{v}{\cdot}), s_v))$,
    \item a tuple $(-a, 0, \bar \dval_s, f(v)+\cost(uv, a, \bar \dval_s))$ for each $a \in \asetplus, \bar\dval_s \in [0, \ieps n/2-1]$.
\end{itemize}

For an internal node $v \in V$ with children $v', v'' \in V$, and for every pair of tuples $(a', \dval'_s, \bar\dval'_s, c') \in \total(v')$, $(a'', \dval''_s, \bar\dval''_s, c'') \in \total(v'')$, we add to $\total(v)$
the tuples corresponding to correct combinations with $(a, \dval_s, \bar\dval_s, c)$, where essentially
$a'+a''+ \indic(b_v) = a$, using the same routing-and-promotion procedure for small demands as
in the single-cable case.

When a new large demand corresponding to a bundle is created, it is rounded as with initial demands (and will often be rounded as a multiple of $\delta_1$ unless $\mu_1$ and $\mu_2$ are close).

For the case of opening a facility in $v$, we still consider any tuples $(a', \dval'_s, 0, c') \in \total(v')$, $(a'', \dval''_s, 0, c'') \in \total(v'')$, $a',a'' \in \asetplus$, as well as any $a \in \asetplus$, $\bar\dval_s \in [0,\ieps n/2-1]$, and add the tuple $(-a, 0, \bar\dval_s, c)$, with cost $c=f(v)+c'+c''+\cost(uv,a,\bar\dval_s)$ to $\total(v)$.

\subsubsection{Computing costs.}
We now specify how to compute $\cost(uv, a, \dval_s)$, for $a \in \asetplus$, $\dval_s \geq 0$.
The general method is similar, since the number of combinations for cable configurations is bounded.

We first compute, for each $i \in [k]$, the set $\aset'_i$ of vectors of capacity at most $\mu_i (1+2\epsilon)$,
\[
\aset'_i = \set{a' \in \asetplus : \sum_{\iota\in [k]} \sum_{j \in [\ieps[2]]} j \delta_\iota \cdot a'_{\iota,j} \leq \mu_i(1+2\epsilon)}.
\]

The sets $\aset'_i$ can be computed explicitly by enumerating all vectors in $\asetplus$. Indeed, $\asetplus$ contains at most
\(
    (n+1)^{k(\ieps[2]-\ieps/2)}
    = n^{O(k\ieps[2])}
\)
vectors.
For each vector $a'\in\asetplus$, we evaluate the left-hand side of the above inequality in $O(k\ieps[2])$ time and include $a'$ in $\aset'_i$ if the inequality is satisfied.
Therefore, the sets $\aset'_i$, for $i\in[k]$, can be computed in time $n^{O(k\ieps[2])}$.

We will compute the table $C$, where $C[a]$ refers to the minimum-cost bag of cables such that each cable can be assigned demand at most $(1+2\epsilon)$ times its capacity, with large demands given by $a$.

To compute $C$, we start with $C[\zeroes]=0$, and then recursively compute
\[
C[a]
= \argmin c\paren*{\set[\big]{
    C[a-a']\cup \set{(\mu_i,c_i)}: i \in [k], a' \in \aset'_i, a-a' \in \asetplus}}
\]

As in the single-cable case, the residual small bundle represented by $\dval_s$ is assigned to any cable in $C[a]$. If $a=\zeroes$ and $\dval_s>0$, we install a cable of minimum cost. This adds less than $\epsilon\mu_1/2\leq\epsilon\mu_i/2$ demand to the selected cable and is accounted for in the resource-augmentation analysis.

\subsubsection{Analysis.}
The analysis follows similarly to the single-cable case.
Let $(I, \mathbf g, \mathbf Q)$ be the solution obtained from running the DP and recovering the best solution.

\begin{claim}[restate=claim:unsplit:kcfl:feasible,name=*]
\label{claim:unsplit:kcfl:feasible}
$(I, \mathbf g, \mathbf Q)$ is a $5\epsilon$-resource augmented solution.
\end{claim}

\thmToAppendix{claim:unsplit:kcfl:feasible}

\begin{appendixproof}
The solution is feasible by construction and the cost of the solution is correctly computed by the dynamic program.
We argue that the solution is $5\epsilon$-resource augmented, as each cable is overloaded by $2\epsilon\mu$ by construction, and then by $2\epsilon\mu$ by the rounding of large demands.

Let a cable of any type $(\mu,c)$ be assigned the demands given by a vector $a \in \asetplus$ with small demands $\lambda_s \in [0,\ieps n/2-1]$.
The rounded demand on the cable can be written as
\[
    \sum_{i\in [k]} \sum_{j \in [\ieps[2]]} j \delta_i \cdot a_{i,j}
        + \lambda_s \delta_s
        = \sum_j \tilde d_j + \lambda_s \delta_s,
\]
where $\tilde d$ is a vector of demands corresponding to $a$, containing $a_{i,j}$
elements of value $j\delta_i$, for every $i \in [k], j \in [\ieps[2]]$.
Let $U \subseteq V$ be the set of vertices whose demands are carried by the cable, and let $U = U_s \uplus U_b$ be the partition of $U$ into small and large demands, respectively.

We remark that for any large demand assigned rounded to size $j\delta_i$, we know that $j \delta_i \leq d_v \leq (j+1)\delta_i \leq j \delta_i (1+ 2\epsilon)$, as $j\delta_i \geq \epsilon\mu_i/2$.

Of these values $\tilde d_j$, some correspond to the grouped small demands of a set $S_j \subseteq U_s$, and so we can write:
\[
\tilde d_j
\leq \sum_{v \in S_j} s_v \delta_s
\leq \tilde d_j + \epsilon^2\mu_1
\leq \tilde d_j (1+2\epsilon),
\]
as we round down to the nearest multiple of $\delta_1 = \epsilon^2\mu_1$,
and $\tilde d_j \geq \epsilon \mu_1/2$.

Similarly, for the demands $\tilde d_j$ corresponding to the large demand of a vertex $u$, we have that
\(
\tilde d_j
\leq d_v
\leq \tilde d_j (1+2\epsilon),
\)
and for the small demands $s_v\delta_s \leq d_v \leq s_v\delta_s + \epsilon^2\mu_1/n$, as a consequence of rounding down to multiples of $\delta_s = \epsilon^2 \mu_1 / n$.

Combining all of these facts, we get that
\begin{align*}
\sum_{v \in U} d_v
&= \sum_{v \in U_s} d_v + \sum_{v \in U_b} d_v \\
&\geq \sum_j \tilde d_j + \lambda_s \delta_s,
\end{align*}

On the other hand, since $\sum_j \tilde d_j \leq \mu (1+2\epsilon)$, we get
\begin{align*}
\sum_{v \in U} d_v
&\leq \sum_{v \in U_s} \paren*{s_v\delta_s+\frac{\epsilon^2\mu}{n}} + \sum_{v \in U_b} \tilde d_v(1+2\epsilon) \\
&\leq \sum_j \tilde d_j(1+2\epsilon) + \lambda_s \delta_s + n \frac{\epsilon^2\mu}{n} \\
&\leq \mu(1+2\epsilon)(1+2\epsilon) + \frac{\epsilon \mu}{2} + \epsilon^2 \mu \\
&\leq \mu(1+5\epsilon),
\end{align*}
where again we use that each $\tilde d_j$ corresponds to either a large demand or bundle of small demands, that the remaining small demands total at most $\epsilon\mu/2$, and that $\epsilon \leq 1/10$.
\end{appendixproof}

\begin{claim}[restate=claim:unsplit:kcfl:optimal,name=*]
\label{claim:unsplit:kcfl:optimal}
$(I, \mathbf g, \mathbf Q)$ has optimal cost for the \kcfl instance,
i.e.\ its cost is at most the cost of an optimum solution not using resource augmentation.
\end{claim}

\thmToAppendix{claim:unsplit:kcfl:optimal}

\begin{appendixproof}
\Cref{claim:unsplit:1cfl:solution} applies for multiple-cable instances, as it only concerns itself with small demands and thus can be applied for cable $(\mu_1, c_1)$.
Thus, there is a solution with optimal cost that uses the rounded-down demands and has total demand on each cable $(\mu,c)$ of $\mu +2\epsilon\mu_1 \leq (1+2\epsilon)\mu$.
Therefore, the proof follows by optimal substructure, as any solutions to a subproblem that have the same demand profile are functionally equivalent.
\end{appendixproof}

\begin{claim}[restate=claim:unsplit:kcfl:time,name=*]
\label{claim:unsplit:kcfl:time}
The algorithm to compute $(I, \mathbf g, \mathbf Q)$ runs in time $n^{O(\ieps[2]k)}$.
\end{claim}

\thmToAppendix{claim:unsplit:kcfl:time}

\begin{appendixproof}
We remark that $\aset$ has size $M = (2n)^{\ieps[2]k} = n^{O(\ieps[2]k)}$.
The size of the dynamic program, the running time of trying all possible combinations for the recursive rules, and the time to compute the installation costs for a given configuration are all polynomial in $M$ and $n$, and thus the running time is $n^{O(\ieps[2])k}$.
\end{appendixproof}

\ifdefined\appendixSection

\section{Deferred Proofs}
\label{sec:proofs}

\prepareAppendixSection
\flushchapterappendix

\fi

\bibliography{ref.bib}
\end{document}